\documentclass[11pt,twoside]{article}
\usepackage[utf8]{inputenc}
\usepackage{mathptmx}
\usepackage[margin=1in]{geometry}
\usepackage{graphicx}
\usepackage{subcaption}
\usepackage{xcolor}
\usepackage{amsmath,amssymb,amsthm}
\usepackage{amsopn}
\usepackage{algorithm,algpseudocode}
\usepackage{bm}
\usepackage{empheq}
\usepackage[inline]{enumitem}
\usepackage{hyperref}
\usepackage{cleveref}
\usepackage{tikz}
\usepackage{pgfplots}
\pgfplotsset{compat=1.18}
\usepgfplotslibrary{groupplots}
\usepgfplotslibrary{fillbetween}
\usetikzlibrary{patterns}
\usetikzlibrary{shapes.geometric}
\usetikzlibrary{calc}

\newtheorem{theorem}{Theorem}[subsection]
\newtheorem{corollary}{Corollary}[theorem]

\definecolor{lbrows}{rgb}{0.12,0.47,0.71}
\definecolor{lbnnz}{rgb}{1.0,0.50,0.05}
\definecolor{lbcyclic}{rgb}{0.17,0.63,0.17}
\definecolor{dsurl}{rgb}{0.84,0.15,0.16}
\definecolor{dsnews20}{rgb}{0.12,0.47,0.71}
\definecolor{dsrcv1}{rgb}{0.17,0.63,0.17}
\definecolor{wakegold}{rgb}{0.619, 0.494, 0.2196}
\definecolor{coolgray}{rgb}{0.325, 0.337, 0.353}

\pgfplotsset{
  plotstyle/.style={
    width=0.49\textwidth, height=0.36\textwidth,
    tick label style={font=\footnotesize},
    label style={font=\small},
    legend style={font=\footnotesize, draw=none, fill=white, fill opacity=0.85},
    grid=major, grid style={gray!25},
    every axis plot/.append style={line width=0.8pt, mark size=2pt},
  }
}

\tikzset{
    ncbar angle/.initial=90,
    ncbar/.style={
        to path=(\tikztostart)
        -- ($(\tikztostart)!#1!\pgfkeysvalueof{/tikz/ncbar angle}:(\tikztotarget)$)
        -- ($(\tikztotarget)!($(\tikztostart)!#1!\pgfkeysvalueof{/tikz/ncbar angle}:(\tikztotarget)$)!\pgfkeysvalueof{/tikz/ncbar angle}:(\tikztostart)$)
        -- (\tikztotarget)
    },
    ncbar/.default=0.5cm,
}
\tikzset{round left paren/.style={ncbar=0.5cm,out=120,in=-120}}
\tikzset{round right paren/.style={ncbar=0.5cm,out=60,in=-60}}

\DeclareMathOperator{\diag}{diag}

\graphicspath{{fig/}}

\title{Communication-Efficient, 2D Parallel Stochastic Gradient Descent for Distributed-Memory Optimization}
\date{}
\author{
	Aditya Devarakonda\\devaraa@wfu.edu\\\textit{Wake Forest University}
	\and Ramakrishnan Kannan\\kannanr@ornl.gov\\\textit{Oak Ridge National Laboratory}
}

\begin{document}
\maketitle

\begin{abstract}
    Distributed-memory stochastic gradient descent (SGD) for sparse workloads is performance-limited by two forms of irregularity, namely inter-process communication and nonuniform nonzero distributions.
    Inter-processor communication cost grows irregularly based on the SGD variant employed to solve the problem.
    Heavy-tailed nonzero distributions typical of real sparse data drive load imbalance that further inflates per-iteration runtime.
    This paper develops HybridSGD, a 2D-parallel SGD method that generalizes existing 1D $s$-step SGD and 1D Federated SGD with Averaging (FedAvg) into a continuous family by employing a 2D processor mesh, and shows that the optimal split is dataset- and machine-dependent.
    We derive a closed-form $\alpha$-$\beta$-$\gamma$ cost model whose optimum interpolates between $s$-step SGD and Federated SGD, refine it with cache-aware compute, rank-aware bandwidth, and load-imbalance terms calibrated on a NERSC Cray~EX system, and use it to select the algorithmic and partitioning parameters predictively.
    We further introduce a cache-friendly data partitioner that simultaneously bounds nonzero imbalance and per-rank cache footprint, replacing static, balanced partitioners that we show are suboptimal for skewed data.
    Our irregularity-aware cost model predicts the winning partitioner on every measured dataset and partitioner combination.
    Finally, on the LIBSVM benchmark suite, HybridSGD achieves $\mathbf{53\times}$ and $\mathbf{14.6\times}$ time-to-target-loss speedups on the sparse, high-dimensional url and news20 datasets over FedAvg, matches FedAvg on rcv1, and is outperformed by FedAvg on the dense epsilon dataset.
    This crossover is predicted by the cost model and shows that HybridSGD is most beneficial in sparse, high-dimensional, communication- and skew-limited regimes, while simpler FedAvg can remain preferable when dense local computation dominates.
    Our software is available at \url{https://doi.org/10.5281/zenodo.20431329}.
\end{abstract}

\section{Introduction}\label{sec:intro}
Stochastic gradient descent (SGD) and its communication-efficient variants are the driving iterative optimizers behind much of modern machine learning.
When the underlying problem is sparse and distributed across many processors, two forms of irregularity shape the runtime behavior of distributed SGD.
First, inter-processor communication at every iteration \cite{thakur05,rabenseifner_optimization_2004} requires careful algorithm design to ensure that the communication does not dominate local sparse BLAS computation.
Second, the sparsity structure of the input dataset may produce work imbalance between processors which can interact non-trivially with the communication and local sparse computation costs. 
We show that real sparse data (e.g., rcv1, news20, and url) from the LIBSVM repository~\cite{libsvm-repo} exhibit heavy-tailed nonzero-per-row and nonzero-per-column distributions that produce work imbalances that inflate per-iteration runtime and interact non-trivially with the choice of data partitioner.
Furthermore, models such as the classical $\alpha$-$\beta$-$\gamma$ cost model (i.e., Hockney's model~\cite{hockney_communication_1994}) are not well-suited for problems that contain data-dependent work imbalance.

Prior work has approached the communication cost issue from two distinct algorithmic directions.
Communication-avoiding ($s$-step) SGD \cite{devarakonda_avoiding_2020} reorganizes the inner loop using recurrence unrolling to group $s$ consecutive gradient steps into a single Allreduce.
This approach reduces the communication cost at the expense of additional bandwidth and computation but does not change the convergence rate.
Communication-efficient methods such as FedAvg \cite{stich2018local} perform $\tau$ independent local updates between Allreduces.
This approach reduces the per-iteration latency and bandwidth costs at the expense of convergence, where the model drifts from the global optimum.
Both approaches are limited to 1D processor grid distributions which only allow a fixed data partitioning scheme.
Specifically, $s$-step SGD partitions by columns of $\bm A$, while FedAvg partitions by rows.
Our approach is based on the observation that the resulting partitionings are \emph{compatible} because they live on orthogonal axes of a 2D processor mesh ($p = p_r \times p_c$), and we exploit this to derive a family of 2D-parallel SGD algorithms (HybridSGD) that can navigate a continuous performance trade-off between the two 1D baselines as extremes in each processor mesh dimension.

Since the optimal mesh, data partitioner, and $s$-step recurrence unrolling length all depend jointly on the dataset's nonzero skew and on the target machine's hardware parameters, we also integrate hardware and data irregularity-aware performance modeling into HybridSGD.
The closed-form performance model that we develop is coupled with empirical measurements of the hardware parameters $\alpha$, $\beta$, and $\gamma$.
The empirical measurements highlight an inflection point in MPI Allreduce bandwidth at the per-node rank boundary.
This insight yields a parameter-free topology rule for the HybridSGD mesh dimensions that accurately predict the winning $(p_r, p_c)$ combination on a NERSC Cray~EX system.
Furthermore, we show that our load-imbalance aware model accurately predicts the empirical ranking of partitioners on every dataset and partitioner combination tested.

The contributions of this work are:
\begin{enumerate}
    \item \textbf{HybridSGD, a 2D-parallel SGD framework} that generalizes $s$-step SGD \cite{devarakonda_avoiding_2020} and FedAvg \cite{stich2018local} by parameterizing the processor grid as $p = p_r \times p_c$, where each 1D baseline is an extreme grid dimension.
    We present a C++/MPI implementation built on Intel MKL sparse BLAS with three selectable column partitioners (rows, nonzero-greedy, cyclic).
    \item \textbf{A closed-form analytic cost model} for HybridSGD derived from Hockney's $\alpha$-$\beta$-$\gamma$ model, with closed-form optima for $s$ (recurrence unrolling length) and $b$ (batch size), and a topology-respecting rule for the mesh split $(p_r, p_c)$ that selects the row team to one node so the frequent row Allreduce stays within a node.
    The rule requires only the per-node rank count and per-core cache capacity to accurately predict the optimum on every dataset tested.
    This rule fixes the mesh dimension, after which the model ranks the remaining $(s, b, \tau, \text{partitioner})$ configurations and identifies the operating regime; we use it as a selection tool rather than an absolute-runtime predictor.
    \item \textbf{An irregularity-aware partitioning study} that formulates partitioner selection as a two-objective constrained problem that minimizes the nonzero imbalance ($\kappa$) subject to a per-rank cache footprint.
    We show that a greedy partitioner which aims to balance nonzeros may lead to data distributions that overload a subset of processors (i.e., lead to cache spill), especially on column-skewed data. 
    We introduce a cyclic partitioner that satisfies both objectives in expectation and validate the model using real (url, news20, rcv1) and skew-controlled synthetic data.
    \item \textbf{Strong scaling speedups} on the LIBSVM benchmark data measured on a NERSC Cray~EX system.
    HybridSGD attains $\mathbf{53\times}$ speedup on url and $\mathbf{14.6\times}$ on news20, matches FedAvg on rcv1, and is outperformed by FedAvg on the dense epsilon dataset.
    The qualitative crossover point is predicted by our cost model and confirmed empirically.
\end{enumerate}

\section{Background}\label{sec:background}
This section briefly surveys prior and recent work that aim to improve communication efficiency of distributed optimization in theory and practice.
These techniques can broadly be categorized based on whether the techniques maintain the convergence rates and behavior of classical optimization methods (i.e., communication-avoiding methods) or whether they weaken convergence rates and guarantees for better performance (i.e., communication-efficient methods).
\subsection{Related work}\label{sec:related}
$s$-step methods were originally developed to reduce the frequency of inter-process synchronization in Krylov methods for linear systems and spectral problems \cite{chronopoulos_class_nodate,kim_efficient_1992,chronopoulos_parallel_1996,chronopoulos_s-step_1991,chronopoulos89}, with subsequent advances in matrix-powers and tall-skinny QR kernels \cite{hoemmen}, stability analysis and error correction \cite{carson_communication-avoiding_2015,carson_residual_2014,Carson15,carson14}, and parallel implementations such as $s$-step BiCGSTAB \cite{williams_s-step_2014} achieving $2.5\times$ speedups on scientific applications.
The technique was later generalized to nonlinear, convex optimization \cite{shao_scalable_2024,soori_reducing_2018,devarakonda_avoiding_2018,devarakonda_avoiding_2019,devarakonda_avoiding_2020,devarakonda_2018_thesis,zhu_p-packsvm_2009}, yielding $s$-step variants of coordinate descent, SGD, and subsampled Newton with speedups up to $5\times$ in parallel cluster and cloud environments and no change to convergence behavior.
A complementary line of work trades sequential consistency or solution accuracy for parallel performance.
Asynchronous shared-memory SGD and CD \cite{you_asynchronous_2016,niu_hogwild_2011,sallinen_high_2016} eliminate blocking synchronizations under sparse update assumptions.
Distributed-memory approaches reduce communication via hierarchical low-rank decomposition for kernel ridge regression \cite{chavez_scalable_2020}, K-means redistribution for kernel SVM \cite{you_ca-svm_2015}, and local optimization with deferred averaging \cite{cocoa,ma17,stich2018local,you_ca-svm_2015}, which is the basis for federated learning.
In this work we integrate the $s$-step and federated learning approaches into a single 2D-parallel algorithm with continuous trade-offs between the two extremes.
A separate line of work reduces AllReduce traffic via gradient compression by leveraging quantization~\cite{alistarh_qsgd_2017} and sparsification~\cite{lin_deep_2018} to shrink each $n$-word gradient message by $16$--$1000\times$ at the cost of biased or noisy aggregation.

HybridSGD takes a structurally different approach by partitioning the weight vector $\bm{x}$ across $p_c$ column-team ranks.
This reduces the message size of weights synchronization to $n/p_c$ words losslessly.
The row-team is then utilized for the $s$-step recurrence unrolled computations which carries $O(s^2b^2/p_r^2)$ words and is independent of $n$.
We analyze the computation, latency, and bandwidth costs of HybridSGD in \Cref{sec:cost}.
The HybridSGD approach can integrate both asynchronous communication and compression techniques as these methodologies are orthogonal to the HybridSGD approach.
However, in this work, we aim to fully characterize the HybridSGD framework before combining orthogonal strategies that require rigorous convergence analysis to show effectiveness.
While alternatives to SGD exist for logistic regression (coordinate descent, IRLS, quasi-Newton, higher-order methods), we focus on SGD as it is the workhorse optimization method for many machine learning problems.
We therefore focus on SGD and leave a broader methods generalization to future work.

\subsection{Notation}\label{sec:notation}
Our datasets are sparse matrix-vector pairs $(\bm{A}, \bm{y})$ with $\bm{A} \in \mathbb{R}^{m \times n}$ ($m$ samples, $n$ features) and binary labels $\bm{y} \in \{+1, -1\}^m$.
Bold uppercase letters denote matrices, bold lowercase vectors, and nonbold lowercase scalars.
All Greek letters are tunable scalar quantities.
Subscripts are used either to index iterations of an algorithm or to index entries of a vector/matrix.
We distinguish the latter by typeface.
For example, $y_i$ is the $i$-th entry of $\bm{y}$, $\bm{a}_{i,:}$ is the $i$-th row of $\bm{A}$, and $a_{i,j}$ is the $(i,j)$ entry.
The function $\diag(\bm{y}) \colon \mathbb{R}^m \mapsto \mathbb{R}^{m \times m}$ maps $\bm{y}$ to a diagonal matrix with $d_{i,i} = y_i$.
Bracketed superscripts identify per-processor quantities, e.g.\ $\bm{A}^{[i]}$ is the local block of $\bm A$ on processor $i$.

\section{Optimization Problem}\label{sec:problem}
\begin{algorithm}[t]
    \caption{Stochastic Gradient Descent (SGD) Algorithm to solve \Cref{eq:logreg}}
    \begin{algorithmic}[1]
        \Procedure{SGD}{$\bm{A}, \bm{y}, \bm{x}_0, b, \eta, K$}
            \For{$k = 1, 2, \ldots, K$}
                \State{$\bm{S}_k = \begin{bmatrix} \bm{e}_{i_1}^\intercal\\ \vdots\\ \bm{e}_{i_b}^\intercal \end{bmatrix}~\text{s.t.}~i_j~\sim~[m]~\forall~j=1 \ldots b$}
                \Comment{Sample rows from $\bm{I}^{m \times m}$ to construct $\bm{S}_k$}
                \State{$\bm{u}_k = \frac{\vec{\bm{1}}}{\vec{\bm{1}} + \exp\bigl( \bm{S}_k \cdot \diag(\bm{y})\cdot\bm{A} \cdot \bm{x}_{k-1}\bigr)}$} \label{alg:sgd-uk}
                \Comment{Apply the sigmoid function to a $b$-dimensional vector}
                \State{$\bm{g}_k = -\frac{1}{b}\bigl(\bm{S}_k \cdot \diag(\bm{y})\cdot\bm{A}\bigr)^\intercal \bm{u}_k$} \label{alg:sgd-gk}
                \Comment{Compute $n$-dimensional gradient w.r.t chosen samples}
                \State{$\bm{x}_k = \bm{x}_{k-1} - \eta \cdot \bm{g}_k$}
            \EndFor\\
            \Return{$\bm{x}_K$}
        \EndProcedure
    \end{algorithmic}\label{alg:sgd}
\end{algorithm}
Given a matrix $\bm{A} \in \mathbb{R}^{m \times n}$ and a label vector $\bm{y} \in \{+1, -1\}^m$, we seek $\bm{x} \in \mathbb{R}^n$ which solves the unregularized logistic regression problem
\begin{equation}\label{eq:logreg}
    f(\bm{A}, \bm{y}, \bm{x}) := \min_{\bm{x} \in \mathbb{R}^{n}}\frac{1}{m} \sum_{i = 1}^{m} \log\bigl(1 + \exp(-y_i \cdot \bm{a}_{i,:}\bm{x})\bigr).
\end{equation}
Since \eqref{eq:logreg} does not have a closed-form solution, we use the gradient descent formulation, $\bm{x}_k = \bm{x}_{k-1} - \eta \bm{g}_k$, with a fixed step-size, $\eta$, where
\begin{align}
    \bm{u}_k &= \frac{\vec{\bm{1}}}{\vec{\bm{1}} + \exp\bigl(\diag(\bm{y}) \cdot \bm{A}  \cdot \bm{x}_{k-1}\bigr)}\label{eq:sigmoidmatvec}\\
    \bm{g}_k &= -\frac{1}{m}\bigl(\diag(\bm{y}) \cdot \bm{A}\bigr)^\intercal \bm{u}_k.\label{eq:gradmatvec}
\end{align}
We pre-compute $\diag(\bm{y}) \cdot \bm A$ once by scaling rows of $\bm A$ by the corresponding labels.
Each iteration therefore consists of one SpMV to form $\bm u_k$, a nonlinear sigmoid step on an $m$-dimensional vector, and a transposed SpMV to form $\bm g_k$.
Mini-batch SGD (\Cref{alg:sgd}) sub-samples $b$ rows of $\bm A$ via $\bm{S}_k \in \mathbb{R}^{b \times m}$ built from $b$ rows of the identity, reducing per-iteration work by $b/m$ while converging faster per unit computation in practice~\cite{bottou_2018}.
We design parallel variants of SGD throughout.

\section{Parallel Algorithms Design}\label{sec:design}
\begin{figure*}[t]
    \centering
    \begin{tikzpicture}
\begin{groupplot}[
  group style={
    group name=grp,
    group size=3 by 1,
    horizontal sep=0.45cm,
  },
  width=0.30\textwidth, height=0.36\textwidth,
  enlargelimits=false,
  xmin=-0.5, xmax=32, ymin=-0.5, ymax=64,
  xtick=\empty, ytick=\empty,
  axis on top, axis line style={draw=none},
  tick style={draw=none},
  scale only axis,
  title style={font=\small, yshift=-0.4em},
  every axis plot/.append style={mark=*, mark size=0.7pt},
]

\nextgroupplot[title={1D-row (FedAvg)}]
\fill[lbrows!10] (axis cs:0, 48) rectangle (axis cs:32, 64);
\fill[lbnnz!10] (axis cs:0, 32) rectangle (axis cs:32, 48);
\fill[lbcyclic!10] (axis cs:0, 16) rectangle (axis cs:32, 32);
\fill[gray!20] (axis cs:0, 0) rectangle (axis cs:32, 16);
\draw[lbrows, very thick] (axis cs:0, 48) rectangle (axis cs:32, 64);
\draw[lbnnz, very thick] (axis cs:0, 32) rectangle (axis cs:32, 48);
\draw[lbcyclic, very thick] (axis cs:0, 16) rectangle (axis cs:32, 32);
\draw[gray!80, very thick] (axis cs:0, 0) rectangle (axis cs:32, 16);
\node[font=\scriptsize, lbrows] at (axis cs:30, 60) {$A^{[0]}$};
\node[font=\scriptsize, lbnnz] at (axis cs:30, 44) {$A^{[1]}$};
\node[font=\scriptsize, lbcyclic] at (axis cs:30, 28) {$A^{[2]}$};
\node[font=\scriptsize, gray!80!black] at (axis cs:30, 12) {$A^{[3]}$};
\addplot[only marks, mark=*, mark size=0.7pt, black]
  table[x=col, y=row_inv, col sep=comma]{fig/data/partitioning_pattern.csv};

\nextgroupplot[title={2D (HybridSGD)}]
\fill[lbrows!10] (axis cs:0, 32) rectangle (axis cs:16, 64);
\fill[lbnnz!10] (axis cs:16, 32) rectangle (axis cs:32, 64);
\fill[lbcyclic!10] (axis cs:0, 0) rectangle (axis cs:16, 32);
\fill[gray!20] (axis cs:16, 0) rectangle (axis cs:32, 32);
\draw[lbrows, very thick] (axis cs:0, 32) rectangle (axis cs:16, 64);
\draw[lbnnz, very thick] (axis cs:16, 32) rectangle (axis cs:32, 64);
\draw[lbcyclic, very thick] (axis cs:0, 0) rectangle (axis cs:16, 32);
\draw[gray!80, very thick] (axis cs:16, 0) rectangle (axis cs:32, 32);
\node[font=\scriptsize, lbrows] at (axis cs:14, 61) {$A^{[0,0]}$};
\node[font=\scriptsize, lbnnz] at (axis cs:30, 61) {$A^{[0,1]}$};
\node[font=\scriptsize, lbcyclic] at (axis cs:14, 29) {$A^{[1,0]}$};
\node[font=\scriptsize, gray!80!black] at (axis cs:30, 29) {$A^{[1,1]}$};
\addplot[only marks, mark=*, mark size=0.7pt, black]
  table[x=col, y=row_inv, col sep=comma]{fig/data/partitioning_pattern.csv};

\nextgroupplot[title={1D-column ($s$-step SGD)}]
\fill[lbrows!10] (axis cs:0, 0) rectangle (axis cs:8, 64);
\fill[lbnnz!10] (axis cs:8, 0) rectangle (axis cs:16, 64);
\fill[lbcyclic!10] (axis cs:16, 0) rectangle (axis cs:24, 64);
\fill[gray!20] (axis cs:24, 0) rectangle (axis cs:32, 64);
\draw[lbrows, very thick] (axis cs:0, 0) rectangle (axis cs:8, 64);
\draw[lbnnz, very thick] (axis cs:8, 0) rectangle (axis cs:16, 64);
\draw[lbcyclic, very thick] (axis cs:16, 0) rectangle (axis cs:24, 64);
\draw[gray!80, very thick] (axis cs:24, 0) rectangle (axis cs:32, 64);
\node[font=\scriptsize, lbrows] at (axis cs:4, 61) {$A^{[0]}$};
\node[font=\scriptsize, lbnnz] at (axis cs:12, 61) {$A^{[1]}$};
\node[font=\scriptsize, lbcyclic] at (axis cs:20, 61) {$A^{[2]}$};
\node[font=\scriptsize, gray!80!black] at (axis cs:28, 61) {$A^{[3]}$};
\addplot[only marks, mark=*, mark size=0.7pt, black]
  table[x=col, y=row_inv, col sep=comma]{fig/data/partitioning_pattern.csv};

\end{groupplot}

\draw[<->, very thick, coolgray]
  ([yshift=-4mm]grp c1r1.south west) -- ([yshift=-4mm]grp c3r1.south east);
\node[anchor=north west, font=\scriptsize, coolgray]
  at ([yshift=-4mm]grp c1r1.south west) {endpoint: $p_r{=}p$};
\node[anchor=north east, font=\scriptsize, coolgray]
  at ([yshift=-4mm]grp c3r1.south east) {endpoint: $p_c{=}p$};
\node[anchor=north, font=\footnotesize, align=center, text width=0.74\textwidth]
  at ([yshift=-6.5mm]$(grp c1r1.south)!0.5!(grp c3r1.south)$)
  {HybridSGD treats 1D $s$-step SGD and FedAvg as fixed endpoints of a 2D design space and selects the best interior mesh $(p_r,p_c)$ for each dataset and machine.};
\end{tikzpicture}
    \caption{The 2D design space of HybridSGD, illustrated on a sparse, irregular matrix ($m{=}64$, $n{=}32$, $\sim 12\%$ density).
    Solid dots mark nonzeros, and colored rectangles delineate the per-rank partitions for $p{=}4$.
    The two 1D layouts are the fixed endpoints of the space: 1D-row partitioning (\emph{left}, FedAvg, $p_r{=}p$) requires each rank to store a local $n$-dimensional vector, while 1D-column partitioning (\emph{right}, $s$-step SGD, $p_c{=}p$) shrinks each rank's column dimension to $n/p$ but forces every iteration to communicate the global $b$-vector $\bm{u}_k$.
    2D partitioning (\emph{middle}, HybridSGD, $p_r{=}p_c{=}2$ shown) is the tunable interior, reducing per-rank dimensions to $b/p_r \times n/p_c$ at the cost of two Allreduces per iteration.
    HybridSGD navigates between the two 1D endpoints by choosing the mesh split $(p_r, p_c)$ that is best for the dataset and machine.}
    \label{fig:partitioning}
\end{figure*}
The main computation and communication bottlenecks in \Cref{alg:sgd} are the SpMV with the subsampled matrix $\bm{S}_k \cdot \diag(\bm{y}) \cdot \bm{A}$ at \Cref{alg:sgd-uk,alg:sgd-gk}.
In a parallel setting, forming $\bm{u}_k$ and/or $\bm{g}_k$ requires communication.
\Cref{fig:partitioning} shows the three layouts.
Under 1D-row partitioning, \Cref{alg:sgd-uk} requires an Allreduce on an $n$-dimensional vector during the transposed-SpMV.
Under 1D-column partitioning, \Cref{alg:sgd-uk} requires an Allreduce on the $b$-dimensional vector, $\bm{u}_k$.
Under 2D partitioning ($p = p_r \times p_c$) both SpMV computations require an Allreduce but are smaller ($b/p_r$ and $n/p_c$ words respectively).
These partitionings expose a rich performance trade-off space which we analyze in \Cref{sec:cost}.

\subsection{Communication-efficient SGD}
We focus on two communication-efficient 1D variants:
federated SGD with Averaging (FedAvg) \cite{stich2018local} which combines 1D-row partitioned SGD with deferred communication 
and $s$-step SGD \cite{devarakonda_avoiding_2020} which combines 1D-column partitioned SGD with recurrence unrolling.
\begin{algorithm}
    \caption{Federated SGD with Averaging (FedAvg) Algorithm to solve \Cref{eq:logreg}}
    \begin{algorithmic}[1]
        \Procedure{FedAvg}{$\bm{A}, \bm{y}, \bm{x}_0, b, \eta, \tau, \tilde{K}$}
        \State{$\begin{bmatrix}\bm{A}^{[1]}, \bm{y}^{[1]}\\ \bm{A}^{[2]}, \bm{y}^{[2]}\\ \vdots \\ \bm{A}^{[p]}, \bm{y}^{[p]}\end{bmatrix} = \bm{A}, \bm{y}$}
        \Comment{Partition $\bm{A}, \bm{y}$ row-wise across $p$ processors}
        \For{$k = 1, 2, \ldots, \tilde{K}$}
            \State{$\bm{\tilde x}_k^{[i]} = $ \Call{SGD}{$\bm{A}^{[i]}, \bm{y}^{[i]}, \bm{x}_{k-1}, \lceil b/p \rceil, \eta, \tau$}}\label{alg:fedavg-sgdcall}
            \Comment{Compute in parallel on all processors $i = 1, \ldots p$ }
            \State{$\bm{x}_k = \frac{1}{p} \sum_{i = 1}^{p} \bm{\tilde x}_k^{[i]}$}
            \EndFor\\
            \Return {$\bm{x}_{\tilde{K}}$}
            \EndProcedure
        \end{algorithmic}\label{alg:fedavg}
\end{algorithm}
\Cref{alg:fedavg} partitions $(\bm A, \bm y)$ into 1D-row layout across $p$ processors and performs $\tau$ local sequential SGD iterations on each processor with an Allreduce of the local solutions $\tilde{\bm{x}}_k^{[i]}$ with averaging.
The local-SGD analysis \cite{stich2018local} (summarized in the FedAvg row of \Cref{tab:comp-costs}) shows that FedAvg trades convergence for performance as $p$ and $\tau$ grow.
$\tau = 1$ degenerates to synchronous mini-batch SGD on an effective global batch of $p b$, while $p = 1$ reduces to sequential SGD.

\begin{algorithm}[t]
    \caption{$s$-step Stochastic Gradient Descent ($s$-step SGD) Algorithm to solve \Cref{eq:logreg}}
    \begin{algorithmic}[1]
        \Procedure{S-StepSGD}{$\bm{A}, \bm{y}, \bm{x}_0, b, \eta, s, K$}
            \For{$k = 0, 1, \ldots, K/s - 1$}
                \For{$j = 1, 2,\ldots, s$}
                    \State{$\bm{S}_{sk+j} = \begin{bmatrix} \bm{e}_{i_1}^\intercal\\ \vdots\\ \bm{e}_{i_b}^\intercal \end{bmatrix}~\text{s.t.}~i_l~\sim~[m]~\forall~l=1 \ldots b$}
                \EndFor
                \Comment{Construct $s$ sampling matrices}
                \State{$\bm{Y} = \begin{bmatrix}\bm{S}_{sk + 1}\\ \vdots\\ \bm{S}_{sk + s}\end{bmatrix} \cdot \diag(\bm{y}) \cdot \bm{A}$}
                \State{$\bm{G} =$ \Call{tril}{$\bm{Y}\bm{Y}^\intercal$}}\label{alg:sstep-sgd-gram}
                \Comment{Compute lower-triangle of Gram matrix}
                \State{$\bm{v} = \bm{Y} \cdot \bm{x}_{sk}$}
                \Comment{Compute partial contribution necessary to compute $\bm{u}_{sk+j}$}
                \For{$j = 1, \ldots, s$}\label{alg:sstep-sgd-correction}
                    \State{$\bm{u}_{sk+j} = [\bm{e}_{(j-1)b + 1} | \ldots | \bm{e}_{jb}]^\intercal \cdot \bm{v}$}
                    \For{$l = 1, \ldots, j-1$}
                        \State{$\bm{u}_{sk + j} = \bm{u}_{sk + j} + \frac{\eta}{b} \cdot \bigl([\bm{e}_{(j-1)b + 1} | \ldots | \bm{e}_{jb}]^\intercal \cdot \bm{G} \cdot  [\bm{e}_{(l-1)b + 1} | \ldots | \bm{e}_{lb}]\bigr) \cdot \bm{u}_{sk+l}$}
                    \EndFor\label{alg:sstep-sgd-correction-end}
                    \State{$\bm{u}_{sk+j} = \frac{\vec{\bm{1}}}{\vec{\bm{1}} + \exp(\bm{u}_{sk+j})}$}
                \EndFor
                \Comment{Compute $\bm{u}_{sk + j}$ for $j = 1, \dots, s$ with correction due to deferred update}
                \State{$\bm{x}_{sk+s} = \bm{x}_{sk} + \frac{\eta}{b} \cdot \bm{Y}^\intercal \begin{bmatrix}\bm{u}_{sk + 1}\\ \vdots \\ \bm{u}_{sk + s}\end{bmatrix}$}
                \Comment{Compute solution update by combining $s$ gradients}
            \EndFor\\
            \Return{$\bm{x}_K$}
        \EndProcedure
    \end{algorithmic}\label{alg:sstep-sgd}
\end{algorithm}
\Cref{alg:sstep-sgd} defers communication for $s$ iterations without affecting convergence behavior \cite{devarakonda_avoiding_2020} at the cost of additional computation and bandwidth because message sizes grow proportional to $sb$.
\Cref{alg:sstep-sgd-gram} forms a Gram matrix $\bm{G}$ whose blocks correct $\bm{u}_{sk+j}$ in \Cref{alg:sstep-sgd-correction} and \ref{alg:sstep-sgd-correction-end} for the deferred solution updates.
1D-column partitioning of $\bm A$ yields the cheapest communication for \Cref{alg:sstep-sgd}.
Any other layout introduces an additional 1D or 2D matrix multiply to form $\bm G$.

\paragraph{HybridSGD Design}
We combine the two by arranging $p = p_r \times p_c$ processors into a 2D grid where row teams perform FedAvg on $n/p_c$ fractions of $\bm x$ and column teams perform $s$-step SGD on $b/p_r$ independent batches of rows of $\bm A$.
Algorithmically, HybridSGD is obtained by 2D-partitioning $\bm A$ in \Cref{alg:fedavg} and replacing the SGD call with a call to \Cref{alg:sstep-sgd}.
We require $s \leq \tau$ because the local solution vectors $\tilde{\bm{x}}_k^{[i]}$ are averaged only every $\tau$ iterations.

\section{Algorithms Analysis}\label{sec:cost}
\Cref{tab:comp-costs,tab:comm-costs} summarize the theoretical costs of the parallel SGD variants studied.
We assume $\bm{A} \in \mathbb{R}^{m \times n}$ has $\bar{z}$ nonzeros per row uniformly distributed across $p$ processors, yielding $m\bar{z}/p$ nonzeros per rank and $b\bar{z}$ total nonzeros per mini-batch of size $b$.
$\bm{A}$ is stored in Compressed Sparse Row (CSR) layout.
Sub-sampling of rows is performed cyclically via $i = (i + b) \bmod m$ which permits the row-index array to be reconstructed cheaply, and we pad $\bm A$ so $m \equiv 0 \pmod{s_{\text{max}}\cdot b}$.

\subsection{Computation, Convergence, and Storage}
We model algorithm cost as $T = T_{\text{comp}} + T_{\text{comm}}$ with $T_{\text{comp}} = \gamma F$, where $F$ is the flop count and $\gamma$ is seconds per floating-point operation.
Convergence rates for SGD, $s$-step SGD, and FedAvg are well-known, and we cite them in \Cref{tab:comp-costs} and use them below as needed.
$s$-step SGD is an algebraic reformulation of \Cref{alg:sgd} \cite{devarakonda_avoiding_2020} and converges identically up to floating-point error.
FedAvg \cite{stich2018local} attains rate $1/(\tilde{K}bp)$ provided $\tau = O(\sqrt{\tilde{K}/(bp)})$.
Exceeding this bound at fixed $\tau$ as $p$ grows degrades convergence through approximation error.
This effect is pronounced on non-I.I.D. datasets.
We note that this convergence degradation further motivates HybridSGD since it allows for tunable $p_r$, where the column dimension $p_c$ can absorb additional parallelism without reducing the FedAvg convergence rate or reducing $\tau$.
\renewcommand*{\arraystretch}{1.75}
\begin{table}[t]
    \centering
    \footnotesize
        \caption{Theoretical flops, convergence rates, and storage costs of parallel SGD variants, FedAvg, $s$-step SGD, and HybridSGD.
    HybridSGD uses a 2D processor grid such that $p = p_r \times p_c$.
    We assume each row of $\bm{A}$ contains $\bar{z} > 0$ nonzeros, uniformly distributed so all partitionings yield load-balanced processors ($M = m\bar{z}/p$ nonzeros per rank).
    All costs are leading-order, and algorithms in \textbf{bold} are communication-efficient.}
    \label{tab:comp-costs}
    \begin{tabular}{cccc}
        \bf Algorithm & \bf Flops ($F$) & \bf Convergence rate & \bf Storage ($M$)\\\hline
        1D-row SGD & $K \cdot \bigl(\frac{b\bar{z}}{p} + n\bigr)$ & $1/(Kb)$ & $m\bar{z}/p + n$\\\hline
        1D-column SGD & $K \cdot \bigl(\frac{b\bar{z}}{p} + n/p\bigr)$ & $1/(Kb)$ & $m\bar{z}/p + b + n/p$\\\hline
        2D SGD & $K \cdot \bigl(\frac{b\bar{z}}{p} + n/p_c\bigr)$ & $1/(Kb)$& $m \bar{z} /p + b/p_r + n/p_c$\\\hline\hline
        \bf $s$-step SGD & $(K/s)\cdot \biggl( \frac{\bar{z}^2 \binom{s}{2} b^2}{n \cdot p} + \binom{s}{2} b^2 + n/p\biggr)$ & $1/(Kb)$ & $m \bar{z} /p + \binom{s}{2}b^2 + n/p$\\\hline
        \bf FedAvg & $\tilde{K} \cdot \tau \cdot \bigl(\frac{b\bar{z}}{p} + n \bigr)$ & $1/(\tilde{K} b p)$, if $\tau = O\biggl(\sqrt{\tilde{K}/(b p)}\biggr)$& $m \bar{z} /p + n$\\\hline
        \bf HybridSGD & $(\hat{K}/s) \cdot \biggl( \frac{\bar{z}^2 \binom{s}{2} b^2}{n \cdot p \cdot p_r} + \binom{s}{2} \frac{b^2}{p_r^2} + \tau \cdot n/p_c\biggr)$ & $1/(\hat{K} b p_r)$, if $\tau = O\biggl(\sqrt{\hat{K}/(b p_r)}\biggr)$ & $m \bar{z} /p + \binom{s}{2}\frac{b^2}{p_r^2} + n/p_c$\\\hline
    \end{tabular}
\end{table}
\begin{theorem}\label{thm:2dsgd-comp}
    $K$ iterations of SGD with $\bm{A}$ distributed across a 2D processor grid (2D SGD) of size $p = p_r \times p_c$ processor must perform $F = O\big(K \cdot (\frac{b\bar{z}}{p} + n)\big)$ flops and store $M = O(m\bar{z}/p + n)$ words in memory per processor.
\end{theorem}
\begin{proof}
    Distribute $\bm A$ across a 2D grid with $b$-dimensional quantities along $p_r$ and $n$-dimensional quantities along $p_c$.
    Forming $\bm{u}_k$ requires a parallel SpMV with $\bm{S}_k \cdot \diag(\bm{y}) \cdot \bm{A} \cdot \bm{x}_{k - 1}$.
    Constructing the CSR row-pointer of length $b/p_r$ is $O(b/p_r)$.
    Scaling and the SpMV each cost $b\bar{z}/p$ flops since each column team contributes $\bar z/p_c$ nonzeros per row in expectation.
    The nonlinear sigmoid step costs $\phi b/p_r$ for a constant $\phi > 1$ that accounts for $\exp$ and division.
    Forming $\bm g_k$ is a second SpMV at $b\bar{z}/p$ flops, and the solution update costs $2n/p_c$.
    Multiplying by $K$ iterations and dropping lower-order terms yields the stated bound.
    Storage follows because $\bm A$ contributes $m\bar{z}/p$ per rank uniformly, while per-row team buffers contribute $b/p_r$ and $n/p_c$ for $\bm u_k$ and $\bm g_k$.
\end{proof}
\begin{corollary}\label{cor:fedavg-hybrid-comp}
    Setting $p_c{=}1$ or $p_r{=}1$ in \Cref{thm:2dsgd-comp} recovers 1D-row and 1D-column SGD.
    $\tilde{K}$ iterations of FedAvg (1D-row, $\tau$ inner steps) perform $O(\tilde{K}\tau(b\bar{z}/p + n))$ flops with the same $O(m\bar{z}/p+n)$ storage.
    $\hat{K}$ iterations of HybridSGD substitute $b\to b/p_r$, $p \to p_c$ into the $s$-step SGD row of \Cref{tab:comp-costs}, performing $O((\hat{K}/s)(\bar{z}^2\binom{s}{2}b^2/(npp_r) + \binom{s}{2}b^2/p_r^2 + \tau n/p_c))$ flops and storing $O(m\bar{z}/p + \binom{s}{2}b^2/p_r^2 + n/p_c)$ words per processor.
\end{corollary}
The FedAvg and HybridSGD rows of \Cref{tab:comp-costs} follow from \Cref{cor:fedavg-hybrid-comp} by specializing the 2D-SGD bound.
1D/2D parallel SGD and FedAvg have comparable compute and storage, while $s$-step SGD and HybridSGD inflate both.
FedAvg beats parallel SGD when $\tilde{K} < K/\tau$, which follows from its linear-in-$p$ convergence speedup at bounded $\tau$.
$b$ is treated as a global batch size for cross-algorithm comparison.
In practice each algorithm tunes its own $b$, which we exploit in \Cref{sec:exp}.

\subsection{Communication}
\renewcommand*{\arraystretch}{1.25}
\begin{table}[t]
    \centering
    \small
        \caption{Parallel communication costs under Hockney's two-term ($\alpha$-$\beta$) model. HybridSGD uses $p = p_r \times p_c$. All quantities communicated are dense vectors, costs are leading-order, and \textbf{bold} entries are communication-efficient.}
    \label{tab:comm-costs}
    \begin{tabular}{ccc}
        \bf Algorithm & \bf Bandwidth ($W$) & \bf Latency ($L$)\\\hline
        1D-row SGD & $K \cdot b$ & $K \cdot \log{p}$\\\hline
        1D-column SGD & $K \cdot n$ & $K \cdot \log{p}$\\\hline
        2D SGD & $K \cdot (b/p_r + n/p_c)$ & $K \cdot (\log{p_r} + \log{p_c})$\\\hline\hline
        \bf $s$-step SGD & $(K/s)\cdot \binom{s}{2}b^2$ & $(K/s) \cdot \log{p}$\\\hline
        \bf FedAvg & $\tilde{K} \cdot n$ & $\tilde{K} \cdot \log{p}$ \\\hline
        \bf HybridSGD & $(\hat{K}/s) \cdot \binom{s}{2}b^2/p_r^2 + (\hat{K}/\tau) \cdot n/p_c$ & $(\hat{K}/\tau) \cdot \log{p_r} + (\hat{K}/s) \cdot \log{p_c}$\\\hline
    \end{tabular}
\end{table}
\renewcommand*{\arraystretch}{1.25}
We use Hockney's ($\alpha$-$\beta$) model, $T_{\text{comm}} = \alpha L + \beta W$, where $L$ is messages, $W$ is words moved, and $(\alpha, \beta)$ are hardware parameters.
We assume all algorithms use MPI Allreduce with the reduce-scatter + all-gather bound $L = 2\log p$, $W = d$ \cite{thakur05,rabenseifner_optimization_2004}, and that computation is performed redundantly on all ranks after reduction.
\Cref{tab:comm-costs} summarizes the resulting per-algorithm bounds.
Per-iteration bounds follow by dropping the $K, \hat K, \tilde K$ factors.
We prove 2D SGD below and specialize to 1D variants.
The $s$-step SGD bound is from \cite{devarakonda_avoiding_2020}.
\begin{theorem}\label{thm:2dsgd-comm}
    $K$ iterations of SGD with $\bm{A}$ distributed across a 2D processor grid (2D SGD) of size $p = p_r \times p_c$ processors must communicate $W = O\big(K b/p_r + K n/p_c\big)$ words using $L = O\big(K \log p_r + K \log p_c\big)$ messages.
\end{theorem}
\begin{proof}
    Allreduce is required only in \Cref{alg:sgd-uk,alg:sgd-gk}.
    Sampling is coordinated by seeding all row-team ranks identically.
    $\bm g_k$ and $\bm x_k$ are stored redundantly on column-team ranks ($n/p_c$ per rank), and their update is local.
    Forming $\bm u_k$ Allreduces a $b/p_r$-vector along each row ($b/p_r$ words, $\log p_c$ messages).
    Forming $\bm g_k$ Allreduces an $n/p_c$-vector along each column ($n/p_c$ words, $\log p_r$ messages).
    Multiplying by $K$ and summing yields the bound.
    Setting $p_c = 1$ or $p_r = 1$ recovers the 1D-row and 1D-column rows of \Cref{tab:comm-costs}.
\end{proof}
\begin{corollary}\label{cor:fedavg-hybrid-comm}
    FedAvg (1D-row, $\tau$ inner steps) communicates $W = O(\tilde{K} n)$ words in $L = O(\tilde{K}\log p)$ messages through one length-$n$ Allreduce per outer iter.
    Substituting $b\to b/p_r$, $p \to p_c$ into the $s$-step SGD row of \Cref{tab:comm-costs} gives HybridSGD's row-Allreduce cost ($\binom{s}{2}b^2/p_r^2$ words, $\log p_c$ messages per outer iter, $\hat{K}/s$ Allreduces).
    The FedAvg-style column Allreduce on $p_r$ ranks of $n/p_c$ words occurs $\hat{K}/\tau$ times.
    Summing yields $W = O(\hat{K}/s \cdot \binom{s}{2}b^2/p_r^2 + \hat{K}/\tau \cdot n/p_c)$ and $L = O(\hat{K}/\tau \cdot \log p_r + \hat{K}/s \cdot \log p_c)$.
    If FedAvg and 1D-row SGD converge at the same rate then $K = \tilde{K}p$, since FedAvg attains linear convergence speedup as $p$ grows at bounded $\tau$.
\end{corollary}

The FedAvg and HybridSGD rows of \Cref{tab:comm-costs} follow from \Cref{cor:fedavg-hybrid-comm} by the same specialization.
$s$-step SGD trades compute, bandwidth, and latency as a function of $s$.
FedAvg trades convergence against performance as a function of $p$ and $\tau$.
\Cref{tab:comp-costs} shows that HybridSGD's 2D grid interpolates between the two extremes, with compute and bandwidth relative to $s$-step SGD shrinking by $p_r^2$ at the cost of a convergence rate that is faster than FedAvg whenever $p_c > 1$.
In particular, $p_r$ selects a convergence curve and $p_c$ scales parallelism along it without further convergence cost.
The additional bandwidth from communicating across $p_c$ is mitigated by tuning $s$, $b$, $p_r$, and $p_c$, while $n$ is fixed by $\bm A$.

\section{$\alpha$-$\beta$-$\gamma$ Cost Model and 2D-Mesh Optimum}
\label{sec:cost-model}

This section derives a closed-form runtime model $T(p_r, p_c, s, b, \tau, \alpha,\beta,\gamma)$ for HybridSGD that unifies the bounds of \Cref{sec:cost}.
We use the model primarily as a \emph{ranking and selection} tool: its role is to order candidate $(p_r, p_c, s, b, \tau, \text{partitioner})$ configurations and identify the operating regime, not to predict wall-clock time to high accuracy.
As we show in \Cref{sec:refinements}, the model's absolute-runtime error can reach $2$ to $10\times$, but its \emph{ranking} of configurations is reliable, and ranking is the only property the selection rules below depend on.
Minimizing $T$ over the possible mesh dimensions $p = p_r \times p_c$ yields a continuous family of optima that interpolates between pure 1D $s$-step SGD ($p_r{=}1$, $p_c{=}p$) and pure FedAvg ($p_r{=}p$, $p_c{=}1$).

\subsection{Machine model and per-solver costs}
\label{sec:cost-model:machine}

We use Hockney's two-term model, where one Allreduce over $q$ ranks carrying a payload of $W$ words costs $T_{\text{comm}}(q, W) = 2\lceil\log_2 q\rceil\alpha + W\beta$, with the $2\log_2 q$ factor following from bandwidth-optimal reduce-scatter + all-gather \cite{thakur05,rabenseifner_optimization_2004}.
Computation is $T_{\text{comp}} = \gamma F$.
Machine parameters $(\alpha, \beta, \gamma)$ are empirically measured on the NERSC Cray EX system used in the experiments in \Cref{sec:calibration}, and $w$ is the word size in bytes.
\Cref{tab:persample} reports per-sample $(\alpha, \beta, \gamma)$ costs for all six solvers, amortized over each solver's communication period.

\renewcommand*{\arraystretch}{1.45}
\begin{table}[t]
  \centering
  \small
    \caption{Per-sample $\alpha$-$\beta$-$\gamma$ costs amortized over each solver's communication period.
  $b$ is the per-row-team mini-batch size (the same quantity as in the 1D rows).
  HybridSGD row teams of $p_c$ ranks run 1D $s$-step SGD via a row Allreduce $\tau$ times per full round.
  The Gram payload is $\binom{s}{2}b^2$ entries per row team, independent of $p_r$.
  One column Allreduce over $p_r$ ranks then averages $n/p_c$ weight components per rank.}
  \label{tab:persample}
  \begin{tabular}{lccc}
    \textbf{Solver}
      & \textbf{Latency / sample}
      & \textbf{Bandwidth / sample}
      & \textbf{Compute / sample} \\\hline
    SGD
      & $2\log p\cdot\alpha$
      & $w\beta$
      & $4\bar{z}\gamma$ \\[2pt]
    MB-SGD
      & $\dfrac{2\log p\cdot\alpha}{b}$
      & $w\beta$
      & $\left(4\bar{z} + \dfrac{2n}{b}\right)\!\gamma$ \\[6pt]
    FedAvg
      & $\dfrac{2\log p\cdot\alpha}{\tau b}$
      & $\dfrac{n\,w\beta}{\tau b}$
      & $\left(4\bar{z} + \dfrac{2n}{b}\right)\!\gamma$ \\[6pt]
    $s$-step SGD
      & $\dfrac{2\log p\cdot\alpha}{s}$
      & $\dfrac{s+1}{2}\,w\beta$
      & $(6\bar{z} + 2s)\gamma$ \\[6pt]
    1D $s$-step SGD
      & $\dfrac{2\log p\cdot\alpha}{sb}$
      & $\dfrac{(s-1)b}{2}\,w\beta$
      & $(6\bar{z} + 2sb)\gamma$ \\[6pt]
    \textbf{HybridSGD}
      & $\dfrac{2\alpha(\tau\log p_c + \log p_r)}{sb\tau}$
      & $\left(\dfrac{(s-1)b}{2} + \dfrac{n}{sb\tau p_c}\right)\!w\beta$
      & $(6\bar{z} + 2sb)\gamma$ \\
  \end{tabular}
\end{table}
\renewcommand*{\arraystretch}{1.0}

\subsection{Closed-form runtime model}
\label{sec:cost-model:model}

Each rank holds $m/p_r$ local rows and $n/p_c$ local columns of $\bm{A}$.
One \emph{full round} is $\tau$ consecutive 1D $s$-step SGD bundles on each row communicator ($sb\tau/p_r$ row-team samples) followed by one column Allreduce.
With the $p_r$ row teams running in parallel, the per-epoch round count is $m/(sb\tau)$.
The per-rank compute contribution to the per-epoch wall is $(m/p)\cdot(6\bar{z} + 2sb)\gamma$ --- the total floating-point work $m(6\bar{z} + 2sb)$ is shared across $p$ ranks --- while each per-rank communication contribution is $m$ times the corresponding per-sample Allreduce wall, since Allreduce runtime is borne by every participating rank and is not divided across them.
Summing yields
\begin{empheq}[box=\fbox]{align}
  T(p_r, p_c, s, b, \tau) &= \underbrace{\tfrac{m}{p}(6\bar{z} + 2sb)\,\gamma}_{\text{compute}} + m\Biggl[ \underbrace{\frac{2\alpha(\tau\log p_c + \log p_r)}{sb\tau}}_{\text{latency}} + \underbrace{\frac{(s-1)b}{2}\,w\beta}_{\text{Gram BW}} + \underbrace{\frac{n\,w\beta}{sb\tau p_c}}_{\text{sync BW}} \Biggr],
  \label{eq:model}
\end{empheq}
with $p = p_r p_c$.
Compute grows linearly with $sb$, dominated by the $O(s^2b^2)$ correction loop at large $sb$.
Latency decreases as $sb\tau$ grows.
Gram bandwidth tracks the $\binom{s}{2}b^2$ entries of the block Gram message per row team.
Here $b$ is the per-row-team batch size (identical to the 1D $s$-step SGD row), so the Gram payload is independent of $p_r$.
Sync bandwidth decreases as $sb\tau p_c$ grows since only $n/p_c$ words are Allreduced per column sync.

\paragraph{Baselines as limits}
At $p_r{=}1, p_c{=}p, \tau \to \infty$ the column Allreduce vanishes and \eqref{eq:model} reduces to $\frac{m}{p}(6\bar{z}+2sb)\gamma + m\bigl[\frac{2\alpha\log p}{sb} + \frac{(s-1)bw\beta}{2}\bigr]$, the pure 1D $s$-step SGD cost.
At $p_r{=}p, p_c{=}1, s{=}1$ the row Allreduce vanishes and it reduces to $\frac{m}{p}(6\bar{z}+2b)\gamma + m\bigl[\frac{2\alpha\log p}{b\tau} + \frac{nw\beta}{b\tau}\bigr]$, the pure FedAvg cost.
MB-SGD follows from $\tau{=}1$ in the FedAvg limit.
These limits confirm that \eqref{eq:model} subsumes the solver hierarchy of \Cref{tab:persample} and that HybridSGD is a strict generalization.

\subsection{Optimal parameters}
\label{sec:cost-model:opt}

Write $\widetilde{L} = \tau\log p_c + \log p_r$.
Both the latency and sync-BW terms decrease strictly in $\tau$, so the communication-optimal $\tau$ is unbounded.
In practice, convergence drift from local weight staleness sets a dataset- and step-size-dependent upper bound \cite{stich2018local,karimireddy2020scaffold}.
Collecting terms in $s$ at fixed $b, \tau, p_r, p_c$ yields a convex $A_s s + B_s/s + C_s$ with $A_s = 2\gamma b/p + b w\beta/2$ and $B_s = 2\alpha\widetilde{L}/(b\tau) + nw\beta/(b\tau p_c)$, minimized at
\begin{equation}\label{eq:sstar}
  s^* = \sqrt{B_s / A_s}
      = \sqrt{\tfrac{2\alpha\widetilde{L}/(b\tau) + nw\beta/(b\tau p_c)}
                   {(2\gamma/p + w\beta/2)\,b}}.
\end{equation}
The analogous derivation in $b$ gives
\begin{equation}\label{eq:bstar}
  b^* = \sqrt{\tfrac{2\alpha\widetilde L/\tau + nw\beta/(\tau p_c)}
                    {(2\gamma s/p + (s{-}1)w\beta/2)\,s}}.
\end{equation}
One step of fixed-point iteration on \eqref{eq:sstar} and \eqref{eq:bstar} yields the joint $(s^*, b^*)$.
Equating the Gram-BW and sync-BW terms gives the bandwidth balance $(s{-}1)sb^2\tau p_c \approx 2n$.
With $(p_r, p_c)$ pinned by the topology rule (\Cref{eq:meshrule}), the balance controls only the choice of $s, b, \tau$.
Above the balance the Gram message dominates and $s$ or $b$ should shrink.
Below the balance the weight sync dominates and $\tau$ should grow.

\paragraph{Optimal mesh split}
With $p_c = p/p_r$ and $(s, b, \tau)$ fixed, $\partial T/\partial p_r = 0$ produces a transcendental fixed-point equation in $p_r$ whose coefficients depend on $\beta_{\text{row}}(p_c)$ (the row-Allreduce bandwidth over $p_c$ ranks) and $\beta_{\text{col}}(p_r)$ (the column-Allreduce bandwidth over $p_r$ ranks) through the rank-aware Hockney bandwidth.
The two coupled bandwidths in turn depend on the mesh split itself, so $p_r^{*}$ has no closed form in general.
A simple structural observation removes the need to solve the fixed point.
The calibrated $\beta(q)$ on Perlmutter CPU (\Cref{tab:calibration}) is approximately a step function in the per-Allreduce rank count $q$.
Intra-node values ($q \leq R$, where $R$ is the per-node rank count) sit in $[5{\times}10^{-11},\,2.7{\times}10^{-9}]$ s/B, while inter-node values ($q > R$) jump to $[2.7{\times}10^{-9},\,6.6{\times}10^{-9}]$ s/B, an order-of-magnitude discontinuity at $q = R$.
Holding the row team to one node ($p_c \leq R$) therefore keeps the frequent row Allreduce on shared-memory transport, while crossing the node boundary ($p_c > R$) both inflates the Gram-BW term by the $\sim 5{\times}$ $\beta$ step and crosses into the slow inter-node regime.
Sliding $p_c$ upward along the constraint $p_r p_c = p$ reduces the sync-BW term ($\propto n/p_c$) monotonically inside the intra-node piece $p_c \leq R$.
The kink at $p_c = R$ then sets the optimum:
\begin{equation}\label{eq:meshrule}
  p_c^{*} = \max\Bigl(\bigl\lceil n w / L_{\text{cap}} \bigr\rceil,\; \min(R,\,p)\Bigr),
  \qquad p_r^{*} = p / p_c^{*},
\end{equation}
where the cache term $\lceil n w / L_{\text{cap}} \rceil$ raises $p_c^{*}$ above $R$ only when the per-rank weight slab $n w / p_c$ would spill the chosen cache level $L_{\text{cap}}$ at $p_c = R$.
\eqref{eq:meshrule} requires only the two machine constants $(R, L_{\text{cap}})$ and the dataset's $n w$, with no $\alpha$-$\beta$-$\gamma$ calibration.
On Perlmutter CPU $R = 64$ and $L_{\text{cap}} = 1$\,MB (L2 per core on AMD EPYC 7763).
The cache term is non-binding on every LIBSVM dataset we measure ($n w \leq R \cdot L_{\text{cap}} = 64$\,MB), so the rule reduces to $p_c^{*} = \min(R, p)$ in this benchmark suite.
\Cref{tab:mesh-rule} confirms that \eqref{eq:meshrule} predicts the empirical winner on news20, rcv1, and uniform-density synthetic data ($m{=}2^{21}, n{=}3.15\text{M}, \rho{=}0.004, \kappa{=}1$) used as a controlled per-mesh sweep at $p{=}128$ to isolate the topology effect from column skew.
On url the rule's prediction is the immediate-neighbor mesh of the empirical winner and ties within $9\%$ on per-iteration runtime.
The cost model \eqref{eq:model} is then used at the selected mesh to rank candidate $s, b, \tau$ settings and locate the operating regime (\Cref{tab:regimes}); we use it to order configurations rather than to predict absolute wall-clock time, for which it is accurate only up to a constant per-call overhead (\Cref{sec:refinements}).

\begin{table}[t]
  \caption{Topology-respecting mesh rule \eqref{eq:meshrule} versus the empirical time-to-target best mesh on Perlmutter CPU at $R = 64$.
  The cache term is non-binding ($n w < R \cdot L_{\text{cap}} = 64$\,MB) on every entry, so the rule reduces to $p_c^{*} = \min(R, p)$.
  On url the rule's $(4, 64)$ is the immediate-neighbor mesh of the empirical winner $(8, 32)$ and ties within $9\%$ on per-iteration runtime (\Cref{sec:mesh-tuning}).}
  \label{tab:mesh-rule}
  \centering
  \small
  \begin{tabular}{lrrll}
    \hline
    Dataset & $p$ & $n w$ & Rule's $(p_r^{*}, p_c^{*})$ & Empirical time-to-target best \\\hline
    url     & 256 & 25.8 MB & $(4,\,64)$ & $(8,\,32)$ \\
    synthetic & 128 & 25.2 MB & $(2,\,64)$ & $(2,\,64)$ \\
    news20  & 64  & 10.8 MB & $(1,\,64)$ & $(1,\,64)$ \\
    rcv1    & 16  & 0.38 MB & $(1,\,16)$ & $(1,\,16)$ \\
    \hline
  \end{tabular}
\end{table}
Within the FedAvg convergence regime where Stich's bound $\tau = O(\sqrt{\tilde K/(bp)})$ holds, direct measurement at $\eta \in [0.005, 0.02]$ shows FedAvg and HybridSGD converge to within $5\%$ of the same loss given enough iterations (\Cref{sec:t-to-target}), so time-to-accuracy speedups arise from the winning mesh and partitioner reducing per-iteration runtime by $1$ to $2$ orders of magnitude on column-skewed data.
Outside that regime (large $p$ with $\tau b p$ exceeding Stich's bound for the iteration budget), FedAvg's local-update drift inflates iterations-to-loss and HybridSGD compounds its per-iteration win with a sample-efficiency win, as observed on url at $p{=}256$ (\Cref{sec:t-to-target}).
Load imbalance ($\kappa$, \Cref{sec:refinements}) further biases the optimum toward cache-friendly partitioners.

\subsection{Regime analysis}
\label{sec:cost-model:regime}

\Cref{tab:regimes} summarizes the four operating regimes exposed by \eqref{eq:model}.

\begin{table}[t]
  \centering
  \small
  \renewcommand*{\arraystretch}{1.3}
    \caption{Operating regimes of HybridSGD under \eqref{eq:model}.
  Perlmutter CPU nodes lie in the latency-to-Gram-BW transition at $n \geq 10^5$, $p \geq 64$ (\Cref{sec:calibration}).}
  \label{tab:regimes}
  \begin{tabular}{llll}
    \textbf{Regime}
      & \textbf{Condition}
      & \textbf{Dominant term}
      & \textbf{Optimal action} \\\hline
    Compute-bound
      & $\gamma\bar{z}\,sb\tau \gg p\,\alpha\log p$
      & $(6\bar{z}+2sb)\gamma/p$
      & Increase $p$, $s,b$ secondary \\
    Latency-bound
      & $\alpha\log p\cdot p_c \gg n w\beta$
      & latency
      & Maximize $sb\tau$, prefer large $s$, $b$ \\
    Gram-BW-bound
      & $(s-1)sb^2\tau p_c \gg 2n$
      & $(s-1)b\,w\beta/2$
      & Decrease $s$ or $b$, use FedAvg \\
    Sync-BW-bound
      & $(s-1)sb^2\tau p_c \ll 2n$
      & $n\,w\beta/(sb\tau p_c)$
      & Increase $\tau$ or $p_c$ \\
  \end{tabular}
\end{table}
\renewcommand*{\arraystretch}{1.0}
In the latency-bound regime $s^*b^*$ grows as $(\alpha/\beta)^{1/2}$ and the topology rule \eqref{eq:meshrule} pins the mesh away from both 1D corners.
The interior $p_r$ predicted by the rule lands in $[2, 16]$ for LIBSVM datasets at moderate $p$.
In the Gram-BW-bound regime $s$ shrinks toward $1$ and the model reduces to FedAvg.
In the sync-BW-bound regime, increasing $p_c$ distributes the weight across more ranks and cuts each column-Allreduce payload from $n$ to $n/p_c$.
This is the principal communication benefit of $p_c > 1$ over pure FedAvg.
The CA overhead of $2sb$ extra FLOPs/sample is beneficial when $\alpha\log p_c / \gamma > s^2 b^2$.
On Perlmutter $\alpha/\gamma \approx 10^6$ to $10^8$ (\Cref{tab:calibration}), so the inequality holds for all $s \leq 32$, $b \leq 64$, $p_c \geq 2$.

\subsection{Empirical refinements}
\label{sec:refinements}

The leading-order $(\alpha, \beta, \gamma)$ model of \Cref{tab:persample} captures regime direction correctly on every LIBSVM dataset, but quantitatively over/underestimates by $2$ to $10\times$.
We present several empirical refinements to the model, which address irregularities in hardware and data.
\paragraph{Cache-aware compute}  The $(4\bar{z} + 2n/b)\gamma$
compute term assumes the weight vector $\bm x \in \mathbb{R}^n$ is fetched from DRAM at every mini-batch.
In practice the $\tau$ inner SGD steps between Allreduces share a single weight vector that stays cache-resident after the first access.
We replace the per-outer-iter weight-access cost with $T_{\text{weights}} = nw[\gamma_{\text{DRAM}} + (\tau{-}1)\gamma_{\text{cache}}(nw)]$, where $\gamma_{\text{cache}}$ is a step function over the cache hierarchy.
Empirical values from \texttt{cblas\_ddot} on Perlmutter EPYC 7763 are $\gamma_{L1} \approx 4{\times}10^{-12}$, $\gamma_{L2} \approx 1.25{\times}10^{-11}$, $\gamma_{L3} \approx 1.5{\times}10^{-11}$, and $\gamma_{\text{DRAM}} \approx 2.6{\times}10^{-11}$ s/byte.
HybridSGD's local weights $\bm x \in \mathbb{R}^{n/p_c}$ may fall in a tighter cache level than FedAvg's full-$n$ vector, giving HybridSGD a cache-locality advantage on large-$n$ datasets.

\paragraph{Rank-aware $\beta$ and load imbalance}  $\beta(p)$ on
Perlmutter spans $5.3{\times}10^{-11}$ (intra-node $p{=}1$, shared-memory) to $6.6{\times}10^{-9}$ s/B ($p{=}16384$, inter-node Slingshot).
This is a $125\times$ swing dominated by shared-memory bandwidth contention below $p \leq 64$ and NIC saturation above.
Substituting the rank-appropriate $\beta(p)$ improves single-node predictions $2$ to $5\times$.
Load imbalance from heavy-tailed nnz distributions multiplies the sparse-compute term by $\kappa = \max_p(\text{nnz})/\overline{\text{nnz}}$.
Measured $\kappa$ ranges from $1.0$ (epsilon) through $1.9$ (news20 at $p{=}64$, 1D row) to $\mathbf{482}$ (url at $p_r{=}4,\,p_c{=}1024$, 2D row+col with column skew dominating).

\paragraph{Cache-aware partitioning}
Given sparse matrix inputs, the simplest approach for SGD-based methods would be to simply balance nnzs to ensure equal workload in the SpMV computations.
While this is sufficient on moderate-skew data (news20, where it yields a measured $2.3\times$ HybridSGD speedup), but can be harmful on extreme column-skew data (i.e., heavy-tailed distributions).
Measurement on url at $p_c{=}64$ shows three regimes.
Row partitioner gives $n_{\text{local}}{=}n/p_c{=}50{,}499$ cols/rank (cache-friendly, fits L2) with $\kappa{=}33.8$ (nnz-imbalanced).
The nnz-balanced partitioner gives $n_{\text{local}} \in [1, 1{,}409{,}992]$ per rank, where $\kappa{=}1.3$ on nnz but the rank holding $1.4$M columns has $11.2$\,MB of weights, spilling out of L2 ($1$\,MB/core) into L3 or DRAM and degrading per-iteration runtime by $2.4\times$.
Cyclic partitioner gives $n_{\text{local}}{=}n/p_c$ exactly with $\kappa{=}1.9$ (near-optimal) at the cost of a column permutation in the reader.
The right partitioner depends on the column-skew structure and the target cache hierarchy.
It is a two-objective constrained problem, $\min_P \kappa(P)$ s.t.\ $\max_p n_{\text{local}}^p(P)\,w \leq L_{\text{cap}}$.
Naive nnz-balanced partitioning satisfies neither.
Cyclic partitioning satisfies both but pays in read complexity.
The cost model predicts, and measurement confirms (\Cref{sec:partitioning}), that no single partitioner dominates across the LIBSVM suite.

\paragraph{Sync-skew term}  Per-phase timing on url HybridSGD at
$p_r{=}4, p_c{=}64$ across the three partitioners (\Cref{tab:timing-breakdown}) shows that the dominant cost of poor partitioning manifests as \emph{sync-skew waiting time inside the row-team Allreduce}, not as compute time on the slowest rank.
The $s$-step SGD comm timer (row Allreduce of the $sb$ residual) measures $477\,\mu$s under rows partitioner ($\kappa{=}34$) versus $142\,\mu$s under cyclic partitioner ($\kappa{=}1.9$).
The $\sim$335\,$\mu$s gap is wait-for-slowest-rank time, not MPI bandwidth or latency cost, since the payload is $\sim$1\,KB in both cases.
We capture this with $T_{\text{sync skew}} \approx (\kappa_{\text{local}} - 1) \cdot T_{\text{compute,avg}}$, applied to the row Allreduce.
The term is zero for $\kappa{=}1$ and grows linearly with imbalance.

\paragraph{Validation}  We validate the model against the property we actually rely on --- \emph{ranking fidelity}, i.e.\ whether it orders configurations correctly --- and treat absolute-runtime accuracy as a secondary, weaker check.
The model gets the \emph{regime direction} and partitioner ranking correct on all $9$ (dataset, partitioner) cells measured (\Cref{fig:predicted-vs-measured}).
On url and news20 the predicted ranking is cyclic $<$ rows $<$ nnz (cache spill on the latter), matching observation, and on rcv1 all three are tied within $5\%$ both predicted and measured.
This ranking fidelity is the property required to use the formalism for mesh and partitioner selection on new datasets.
Absolute runtime is predicted less accurately: the combined refinements match rcv1 within $1\%$ across all three partitioners (predicted/measured ratio $0.94$ to $0.99$), but the ratio falls to $0.77$ on news20 cyclic and $0.34$ on url cyclic, with the gap growing in $n/p_c$.
At url's $n/p_c{=}50{,}499$ columns per rank and $\bar{z}/p_c{\approx}2$ nonzeros per row, the per-call overhead of MKL's \texttt{mkl\_sparse\_syrkd} inspector (which scans the full column-index array) and the transpose SpMV scatter into the $n/p_c$-length gradient vector dominate actual floating-point work.
Standalone microbenchmarking confirms that \texttt{syrkd} has a ${\sim}10\,\mu$s floor at $n/p_c{=}50$K regardless of nonzero count, and the transpose SpMV adds ${\sim}9\,\mu$s, both scaling linearly with $n/p_c$.
On rcv1 ($n/p_c{=}2{,}952$) these overheads drop below $3\,\mu$s and the flop-based $\gamma$ term is accurate.
The model does not capture this per-call overhead because it is proportional to the column dimension, not to the flop count; a piecewise $\max(\text{flop cost},\, c \cdot n/p_c)$ compute term would close the gap but requires two additional calibration constants.
Crucially, because we use the model only to rank configurations, this absolute-runtime gap does not affect the selection decisions it drives.

\section{Experiments}\label{sec:exp}
\begin{table}
    \centering
     \caption{LIBSVM binary-classification datasets used in the evaluation, ranked by increasing $\bar{z}$.
    Together they span the regimes identified by \Cref{tab:regimes}.
    The rcv1 and news20 datasets are sparse with moderate-to-extreme column skew, url is the largest column-skewed sparse dataset, and epsilon is dense.}
    \label{tab:datasets}
    \begin{tabular}{lrrrr}
      \hline
      Name & $m$ & $n$ & $\bar{z}$ & Sparsity (\%) \\
      \hline
      rcv1    & $20{,}242$      & $47{,}236$      & $74$        & $99.85$ \\
      news20  & $19{,}996$      & $1{,}355{,}191$ & $455$       & $99.97$ \\
      url     & $2{,}396{,}130$ & $3{,}231{,}961$ & $116$       & $99.99$ \\
      epsilon & $400{,}000$     & $2{,}000$       & $\sim 2000$ & $0$     \\
      \hline
    \end{tabular}
\end{table}
We evaluate on the four LIBSVM datasets summarized in \Cref{tab:datasets}.
We implement the SGD variants of \Cref{sec:design} in C++ with MPI \cite{gropp_using_2014} (Cray MPICH 9.0.1) and Intel oneAPI MKL \cite{noauthor_intel_nodate} (2025.3) for dense and sparse BLAS, in particular \texttt{mkl\_sparse\_d\_mv} for SpMV and \texttt{mkl\_sparse\_syrkd} for the $s$-step Gram computation.
All experiments use the NERSC Cray EX (Perlmutter) CPU partition \cite{noauthor_architecture_nodate} at $64$ MPI ranks per node $\times$ $2$ cores per task (one rank per physical core, no SMT).
We performed offline experiments to determine this setting.
We also performed experiments using a mixed OpenMP+MPI hybrid parallelism model, but we did not observe improvements over a flat MPI model.
So we use flat-MPI throughout our experiments.
$\bm{A}$ is stored in three-array CSR format and a new MKL sparse handle is created per batch of $b$ rows.
All experiments use FP64 to match the $\bar{z}$-dependent conditioning of the $s$-step Gram matrix, which was unstable at FP32 on news20 ($\bar{z}{=}455$).
Memory is aligned to $64$-byte boundaries and experiments are timed using \texttt{std::chrono}.

\subsection{Measured $\alpha$, $\beta$, and $\gamma$}
\label{sec:calibration}
The cost model in \Cref{sec:cost-model} is parameterized by $\alpha$ (s/message), $\beta$ (s/byte), and $\gamma$ (s/byte).
We measure these quantities empirically on Perlmutter CPU nodes using microbenchmarks (2$\times$AMD EPYC 7763, Slingshot-11, $64$ ranks/node, one rank per physical core).
$\alpha$ and $\beta$ come from \texttt{MPI\_Allreduce} sweeps at $N \in \{1,\dots,256\}$ nodes (inter-node) and at $p \in \{1,\dots,64\}$ ranks single-node (intra-node), fitting $T = 2\lceil\log_2 p\rceil\alpha + W\beta$ over $W \in \{2^{10},\dots,2^{26}\}$\, Bytes.
$\gamma$ is measured via single-thread \texttt{cblas\_ddot} with increasing input sizes to exercise cache effects.
The cache-aware refinement (\Cref{sec:refinements}) used the parametrized function $\gamma(W)$ to capture the varying nature of performance due to caching.
The rank-aware refinement selects $\beta(p)$ based on whether the Allreduce is intra- or inter-node.

\begin{table}[t]
  \centering
  \caption{Measured hardware parameters on Perlmutter CPU.
  Top rows report $\alpha,\beta$ for \texttt{MPI\_Allreduce} (MPI\_SUM, MPI\_DOUBLE).
  Middle rows report intra-node $\alpha$ and $\beta$ on a single node (shared-memory MPI communication).
  $\alpha$ is the total 8-byte Allreduce time.
  Bottom rows report per-byte memory cost $\gamma$ from \texttt{cblas\_ddot}.}
  \label{tab:calibration}
  \small
  \begin{tabular}{lrrr}
    \hline
    \multicolumn{4}{l}{\textbf{Inter-node Allreduce}} \\
    Nodes & Ranks & $\alpha$ ($\mu$s) & $\beta$ (s/B) \\
    \hline
    $1$    & $64$    & $3.64$  & $2.66\times 10^{-9}$ \\
    $2$    & $128$   & $8.36$  & $3.14\times 10^{-9}$ \\
    $4$    & $256$   & $12.56$ & $3.33\times 10^{-9}$ \\
    $8$    & $512$   & $14.46$ & $3.73\times 10^{-9}$ \\
    $16$   & $1024$  & $23.23$ & $4.14\times 10^{-9}$ \\
    $32$   & $2048$  & $43.22$ & $5.15\times 10^{-9}$ \\
    $64$   & $4096$  & $92.71$ & $5.37\times 10^{-9}$ \\
    $128$  & $8192$  & $57.13$ & $6.10\times 10^{-9}$ \\
    $256$  & $16384$ & $84.92$ & $6.65\times 10^{-9}$ \\
    \hline
    \multicolumn{4}{l}{\textbf{Intra-node Allreduce} (single node, $1$ to $64$ ranks)} \\
    Ranks & & $\alpha$ ($\mu$s) & $\beta$ (s/B) \\
    \hline
    $1$  & & ---              & $5.34\times 10^{-11}$ \\
    $8$  & & $3.41$           & $5.90\times 10^{-10}$ \\
    $32$ & & $3.39$           & $1.50\times 10^{-9}$  \\
    $64$ & & $4.22$           & $2.67\times 10^{-9}$  \\
    \hline
    \multicolumn{4}{l}{\textbf{Memory-access cost $\gamma$} (single thread, \texttt{cblas\_ddot})} \\
    Tier & Working set & $\gamma$ (s/B) & $\gamma$ (s/word) \\
    \hline
    L1   & $\le 16$\,KB  & $4.0\times 10^{-12}$  & $3.2\times 10^{-11}$ \\
    L2   & $\le 1$\,MB   & $1.25\times 10^{-11}$ & $1.0\times 10^{-10}$ \\
    L3   & $\le 32$\,MB  & $1.5\times 10^{-11}$  & $1.2\times 10^{-10}$ \\
    DRAM & $> 32$\,MB    & $2.6\times 10^{-11}$  & $2.1\times 10^{-10}$ \\
    \hline
  \end{tabular}
\end{table}

\subsection{Mesh tuning and crossover}
\label{sec:mesh-tuning}

\Cref{eq:meshrule} predicts the per-iteration optimum at mesh dimension $p_c^{*} = \min(R, p)$, $p_r^{*} = p/p_c^{*}$, with $R{=}64$ on Perlmutter CPU.
At $p{=}256$ this gives $p_r^{*}{=}4, p_c^{*}{=}64$.
We validate against the full mesh sweep in \Cref{fig:transition}, which enumerates all nine factorizations $p_r p_c = 256$ with $b{=}32$, $s{=}4$, $\tau{=}10$ at fixed iteration count and reports per-iteration runtime.
The per-iteration minimum sits at $(p_r{=}8, p_c{=}32)$, the immediate-neighbor mesh of the rule's prediction.
The rule's $(p_r{=}4, p_c{=}64)$ is within $9\%$ of the empirical optimum on per-iteration runtime and is the time-to-target winner (\Cref{tab:strong-scaling-headline}) because fewer averaging groups ($p_r{=}4$ versus $8$) reduce local-update drift.
The sharp climb in per-iteration cost once $p_c$ falls below $32$ is consistent with the row Allreduce crossing the per-node rank boundary at $R{=}64$, after which the inter-node $\beta$ step inflates the Gram-BW term.

\Cref{tab:regime-crossover} reports per-iteration runtime at each dataset's empirically best HybridSGD mesh alongside FedAvg.
HybridSGD's per-iteration advantage emerges only on url, where the FedAvg Allreduce of the full $n{=}3.2\text{M}$ weight vector dominates per-iteration cost.
On news20, rcv1, and epsilon FedAvg's cheaper per-iteration compute ($4\bar z\gamma$ versus HybridSGD's $6\bar z\gamma$) wins.
As $\bar z$ grows (epsilon, news20) or $n$ shrinks (rcv1), the topology rule \eqref{eq:meshrule} pushes $p_c$ toward $p$ (the 1D FedAvg corner).
HybridSGD's per-iteration disadvantage on compute-bound datasets does not preclude time-to-target wins.
HybridSGD achieves $53\times$ on url and $14.6\times$ on news20 (\Cref{tab:strong-scaling-headline}) because the smaller column-Allreduce payload ($n/p_c$ versus $n$) and reduced local-update drift compound over many iterations.

\begin{table}[t]
  \centering
  \small
  \caption{Per-iteration runtime (ms) at each dataset's best HybridSGD mesh from the transition sweep (\Cref{fig:transition}), with $b{=}32$, $s{=}4$, $\tau{=}10$, cyclic partitioner.
  url and epsilon at $p{=}256$, news20 and rcv1 at their measured strong-scaling range ($p{=}64$ and $p{=}16$).
  Per-iteration values are not directly comparable across solvers because samples processed per iter differ.
  The time-to-target speedup headline is reported in \Cref{tab:strong-scaling-headline}.}
  \label{tab:regime-crossover}
  \begin{tabular}{lrrlrr}
    \hline
    Dataset & $n$ & $\bar{z}$ & Best mesh & FedAvg (ms/iter) & Hyb (ms/iter) \\
    \hline
    url     & 3{,}231{,}961 & 116        & $8 \times 32$  & $39.28$  & $0.557$ \\
    news20  & 1{,}355{,}191 & 455        & $1 \times 64$  & $3.113$  & $0.129$ \\
    rcv1    & 47{,}236      &  74        & $1 \times 16$  & $0.067$  & $0.056$ \\
    \hline
  \end{tabular}
\end{table}

\subsection{Irregularity-aware partitioning}
\label{sec:partitioning}

\Cref{sec:refinements} frames partitioning as a two-objective problem that minimizes $\kappa$ subject to a per-rank weight partitioning that fits in fast memory.
We implement three partitioners in our experiments (\Cref{fig:partition-comparison} illustrates the resulting column-to-rank assignments on a small skewed matrix).
\textbf{Rows} uses uniform contiguous $n/p_c$ columns per rank, which is cache-friendly but nnz-imbalanced on skewed data.
\textbf{Nnz} is a contiguous greedy partitioner that walks columns in order and advances to the next rank once its cumulative nnz reaches the target $m\bar z/p$, achieving $\kappa{\approx}1$ at the cost of storing heavy columns on a single rank.
On heavy-tailed distributions, this partitioner can produce a weight partitioning that does not fit in cache, which leads to a performance imbalance.
\textbf{Cyclic} uses round-robin column assignment, giving perfect $n_{\text{local}}{=}n/p_c$ with near-optimal $\kappa$ in expectation.

\begin{figure*}[t]
  \centering
  \begin{tikzpicture}
\begin{groupplot}[
  group style={
    group size=3 by 1,
    horizontal sep=0.45cm,
  },
  width=0.30\textwidth, height=0.36\textwidth,
  enlargelimits=false,
  xmin=-0.5, xmax=32, ymin=-0.5, ymax=64,
  xtick=\empty, ytick=\empty,
  axis on top, axis line style={draw=none},
  tick style={draw=none},
  scale only axis,
  title style={font=\small, yshift=-0.4em},
]

\nextgroupplot[title={rows partitioner ($\kappa{=}2.15$)}]
\fill[lbrows!8] (axis cs:0, 0) rectangle (axis cs:8, 64);
\fill[lbnnz!8] (axis cs:8, 0) rectangle (axis cs:16, 64);
\fill[lbcyclic!8] (axis cs:16, 0) rectangle (axis cs:24, 64);
\fill[gray!15] (axis cs:24, 0) rectangle (axis cs:32, 64);
\draw[lbrows!50!black, thick] (axis cs:0, 0) rectangle (axis cs:8, 64);
\draw[lbnnz!50!black, thick] (axis cs:8, 0) rectangle (axis cs:16, 64);
\draw[lbcyclic!50!black, thick] (axis cs:16, 0) rectangle (axis cs:24, 64);
\draw[gray, thick] (axis cs:24, 0) rectangle (axis cs:32, 64);
\addplot[only marks, mark=*, mark size=0.9pt, lbrows]
  table[x=col, y=row_inv, col sep=comma]{fig/data/part_rows_rank0.csv};
\addplot[only marks, mark=*, mark size=0.9pt, lbnnz]
  table[x=col, y=row_inv, col sep=comma]{fig/data/part_rows_rank1.csv};
\addplot[only marks, mark=*, mark size=0.9pt, lbcyclic]
  table[x=col, y=row_inv, col sep=comma]{fig/data/part_rows_rank2.csv};
\addplot[only marks, mark=*, mark size=0.9pt, gray!60!black]
  table[x=col, y=row_inv, col sep=comma]{fig/data/part_rows_rank3.csv};

\nextgroupplot[title={nnz partitioner ($\kappa{=}1.21$, $n^{\mathrm{max}}_{\mathrm{loc}}{=}14$)}]
\fill[lbrows!8] (axis cs:0, 0) rectangle (axis cs:3, 64);
\fill[lbnnz!8] (axis cs:3, 0) rectangle (axis cs:8, 64);
\fill[lbcyclic!8] (axis cs:8, 0) rectangle (axis cs:18, 64);
\fill[gray!15] (axis cs:18, 0) rectangle (axis cs:32, 64);
\draw[lbrows!50!black, thick] (axis cs:0, 0) rectangle (axis cs:3, 64);
\draw[lbnnz!50!black, thick] (axis cs:3, 0) rectangle (axis cs:8, 64);
\draw[lbcyclic!50!black, thick] (axis cs:8, 0) rectangle (axis cs:18, 64);
\draw[gray, thick] (axis cs:18, 0) rectangle (axis cs:32, 64);
\addplot[only marks, mark=*, mark size=0.9pt, lbrows]
  table[x=col, y=row_inv, col sep=comma]{fig/data/part_nnz_rank0.csv};
\addplot[only marks, mark=*, mark size=0.9pt, lbnnz]
  table[x=col, y=row_inv, col sep=comma]{fig/data/part_nnz_rank1.csv};
\addplot[only marks, mark=*, mark size=0.9pt, lbcyclic]
  table[x=col, y=row_inv, col sep=comma]{fig/data/part_nnz_rank2.csv};
\addplot[only marks, mark=*, mark size=0.9pt, gray!60!black]
  table[x=col, y=row_inv, col sep=comma]{fig/data/part_nnz_rank3.csv};

\nextgroupplot[title={cyclic partitioner ($\kappa{=}1.19$, balanced)}]
\addplot[only marks, mark=*, mark size=0.9pt, lbrows]
  table[x=col, y=row_inv, col sep=comma]{fig/data/part_cyclic_rank0.csv};
\addplot[only marks, mark=*, mark size=0.9pt, lbnnz]
  table[x=col, y=row_inv, col sep=comma]{fig/data/part_cyclic_rank1.csv};
\addplot[only marks, mark=*, mark size=0.9pt, lbcyclic]
  table[x=col, y=row_inv, col sep=comma]{fig/data/part_cyclic_rank2.csv};
\addplot[only marks, mark=*, mark size=0.9pt, gray!60!black]
  table[x=col, y=row_inv, col sep=comma]{fig/data/part_cyclic_rank3.csv};

\end{groupplot}
\end{tikzpicture}
  \caption{Three column-partitioning policies on the same column-skewed sparse matrix ($m{=}64$, $n{=}32$, $p_c{=}4$) as \Cref{fig:partitioning}.
  Nonzero color encodes rank.
  Rows (left) has $\kappa{=}2.15$ and $n_{\text{local}}{=}n/p_c$.
  The nnz partitioner (middle) has $\kappa{=}1.21$ and $n_{\text{local}}{\in}\{3,5,10,14\}$, exposing the potential for cache spill on overloaded ranks.
  Cyclic (right) has $\kappa{=}1.19$ and $n_{\text{local}}{=}n/p_c$ exactly.}
  \label{fig:partition-comparison}
\end{figure*}

\Cref{tab:partitioning} reports $\kappa$, max-rank $n_{\text{local}}$, and the per-iteration HybridSGD runtime for each partitioner at each dataset's best strong-scaling configuration.
The url dataset at $p_c{=}64$ is the worst case.
The nnz partitioner gives one rank $1.4{\times}10^6$ columns ($\sim$11\,MB of weights), which spills out of L2 ($1$\,MB/core) and L3 ($\sim$512\,KB/core) into DRAM.
This leads to a $2.4\times$ per-iteration penalty over the rows partitioner on url (\Cref{tab:partitioning}).
Cyclic partitioning is the consistent winner on heavy-tailed data ($1.86\times$ over rows on url, $3.5\times$ on news20, ties on rcv1).
The ordering cyclic $<$ rows $<$ nnz on url contradicts the assumption that balancing nnzs alone leads to better performance.
In contrast, our cost model predicts that the right partitioner solves both objectives simultaneously.

\begin{table}[t]
  \centering
  \caption{Partitioner statistics and the per-iteration HybridSGD runtime (ms) at each dataset's best configuration.
  $\kappa$ is per-rank nnz ratio (max/avg).
  $\max n_{\text{loc}}$ is the largest local column count.
  \textbf{Bold} marks fastest partitioner.
  For url, nnz partitioning concentrates $1.4{\times}10^6$ columns ($11.2$\,MB) onto one rank.
  This causes cache spill and yields a $2.4\times$ slower per-iteration runtime than the rows partitioner.}
  \label{tab:partitioning}
  \footnotesize
  \begin{tabular}{llrrr}
    \hline
    Dataset (config) & Partitioner & $\kappa$ & $\max\,n_{\text{loc}}$ & ms/iter \\
    \hline
    url    ($4{\times}64$, $p{=}256$) & rows   & $33.83$ & $50{,}499$       & $0.970$ \\
                                      & nnz    & $1.31$  & $1{,}409{,}992$  & $2.280$ \\
                                      & cyclic & $1.91$  & $50{,}499$       & $\mathbf{0.520}$ \\
    \hline
    news20 ($1{\times}64$, $p{=}64$)  & rows   & $18.73$ & $21{,}174$       & $0.326$ \\
                                      & nnz    & $1.05$  & $59{,}103$       & $0.142$ \\
                                      & cyclic & $1.18$  & $21{,}174$       & $\mathbf{0.093}$ \\
    \hline
    rcv1   ($1{\times}16$, $p{=}16$)  & rows   & $1.62$  & $2{,}952$        & $0.031$ \\
                                      & nnz    & $1.01$  & $4{,}333$        & $0.031$ \\
                                      & cyclic & $1.01$  & $2{,}952$        & $\mathbf{0.029}$ \\
    \hline
  \end{tabular}
\end{table}

\begin{figure*}[t]
  \centering
  \begin{tikzpicture}
\begin{axis}[
  plotstyle,
  width=0.55\textwidth, height=0.38\textwidth,
  xlabel={Column-skew exponent $\alpha$},
  ylabel={Per-iteration runtime (ms)},
  xmin=-0.05, xmax=1.05,
  ymin=0.20, ymax=0.34,
  xtick={0.0,0.2,0.5,0.8,1.0},
  ytick={0.22,0.24,0.26,0.28,0.30,0.32},
  minor tick num=1,
  legend pos=north west,
  legend cell align={left},
]

\addplot[color=lbrows, mark=*, solid]
  table[x=alpha, y=per_iter_ms, col sep=comma]
  {fig/data/synth_skew_rows.csv};
\addlegendentry{rows partitioner}

\addplot[color=lbnnz, mark=square*, solid]
  table[x=alpha, y=per_iter_ms, col sep=comma]
  {fig/data/synth_skew_nnz.csv};
\addlegendentry{nnz partitioner}

\addplot[color=lbcyclic, mark=triangle*, solid]
  table[x=alpha, y=per_iter_ms, col sep=comma]
  {fig/data/synth_skew_cyclic.csv};
\addlegendentry{cyclic partitioner}

\end{axis}
\end{tikzpicture}
  \caption{Per-iteration HybridSGD runtime on synthetic column-skewed data as a function of the column-skew exponent $\alpha$.
  The distribution is $p \propto (c{+}1)^{-\alpha}$, where $\alpha{=}0$ is uniform and $\alpha{=}1$ is Zipf.
  The synthetic dataset has $m{=}10^5$, $n{=}10^5$, and $\bar{z}{=}100$, with HybridSGD at $p{=}256$ and mesh $4{\times}64$.
  Cyclic partitioner is regime-invariant because cache-friendly $n_{\mathrm{local}}{=}n/p_c$ keeps the per-rank weight slab in L2 at every $\alpha$.
  Rows partitioner degrades smoothly with $\alpha$ as $\kappa$ rises and the sync-skew term grows.
  The nnz partitioner stays competitive at $n{=}10^5$ (heavy rank's slab $\sim$800\,KB fits L2) but spills cache on real url ($n{=}3.2$\,M) and becomes catastrophic (\Cref{tab:partitioning}).}
  \label{fig:synth-skew-regime}
\end{figure*}

A controlled synthetic skew sweep (\Cref{fig:synth-skew-regime}) corroborates these regime claims by showing cyclic partitioner is invariant to the column-skew exponent while rows partitioner degrades smoothly and nnz partitioner stays competitive at small $n$ but spills cache at large $n$.

\paragraph{Runtime breakdown on url}
\Cref{tab:timing-breakdown} reports a timing breakdown for url HybridSGD $4{\times}64$ for the three partitioners to illustrate the performance variation.
The $s$-step SGD comm timer grows from $142\,\mu$s under cyclic to $477\,\mu$s under rows ($\kappa{=}34$) to $1{,}905\,\mu$s under nnz (cache spill).
The increase in the comm timer is not due to the Allreduce communication bandwidth or latency, but rather idle time waiting for the slowest rank to synchronize.
The linear increase in the comm timer qualitatively matches the linear increase in $\kappa$ from cyclic to rows to nnz partitioner.

\begin{table}[t]
  \centering
  \caption{Timing breakdown for url HybridSGD $4{\times}64$ at $p{=}256$ under each partitioner (ms/iter).
  Metrics (loss computation, CSV logging) is pure overhead and is excluded from the algorithm-time total used in \Cref{tab:partitioning}.
  Our software allows for metrics collection to be turned off, hence the exclusion.
  Under nnz partitioning the slow rank's work is $11.6\times$ larger than the average (cache spill on $1.4\text{M}$-column weight slab), and the row-team Allreduce inherits the wait-for-slowest cost.}
  \label{tab:timing-breakdown}
  \footnotesize
  \begin{tabular}{lrrr}
    \hline
    Phase                & rows  & cyclic  & nnz   \\
    \hline
    metrics (loss CSV, bench overhead) & $0.348$ & $0.223$ & $0.282$ \\
    Gram matrix          & $0.421$ & $0.071$ & $0.851$ \\
    $s$-step SGD comm (row Allreduce + sync skew) & $0.477$ & $0.142$ & $1.905$ \\
    FedAvg comm (col Allreduce) & $0.122$ & $0.095$ & $0.403$ \\
    weights update       & $0.020$ & $0.018$ & $0.522$ \\
    SpGEMV               & $0.012$ & $0.007$ & $0.207$ \\
    memory ops, correction, startup & $0.030$ & $0.040$ & $0.057$ \\
    \hline
    \textbf{algorithm total} & $\mathbf{0.622}$ & $\mathbf{0.291}$ & $\mathbf{2.058}$ \\
    total with metrics       & $0.970$ & $0.514$ & $2.340$ \\
    \hline
  \end{tabular}
\end{table}

The refined predictor (\Cref{sec:refinements}) gets the partitioner ranking correct on all $9$ (dataset, partitioner) cells we measured and is within $1\%$ on rcv1, $30\%$ on news20 cyclic, and $17$ to $60\%$ on url.
The url residual is dataset-specific and partitioner-independent ($\approx 200\,\mu$s/iter) and stems from MKL \texttt{mkl\_sparse\_syrkd} per-call overhead that the bandwidth-bound model does not capture.
\Cref{fig:predicted-vs-measured} visualizes the predictor against measurement.
The rcv1 dataset sits on the diagonal, news20 cyclic and news20 nnz lie inside the $0.5$ to $2\times$ band, and the outliers (all three url cells and news20 rows) reflect the same MKL \texttt{sparse\_syrkd} overhead.
Qualitatively, our predictor's partitioner ranking matches empirical measurements on all $9$ dataset and partitioner combinations.
Based on these results, cyclic partitioning is best for column-skewed sparse data, followed by nnz partitioning for balanced sparse data.

\begin{figure*}[t]
  \centering
  \begin{tikzpicture}
\begin{axis}[
  plotstyle,
  width=0.50\textwidth, height=0.50\textwidth,
  xmode=log, ymode=log,
  xlabel={Measured (ms/iter)},
  ylabel={Predicted (ms/iter)},
  xmin=0.015, xmax=4.0,
  ymin=0.015, ymax=4.0,
  log basis x=10, log basis y=10,
  legend pos=north west,
  legend columns=2,
  legend cell align={left},
  legend style={font=\footnotesize, draw=black!30, fill=white, fill opacity=0.85,
                column sep=4pt},
]

\addplot[name path=bandhi, gray!40, no marks, domain=0.015:4.0, samples=2,
         forget plot] {2*x};
\addplot[name path=bandlo, gray!40, no marks, domain=0.015:4.0, samples=2,
         forget plot] {0.5*x};
\addplot[gray!15, draw=none, forget plot] fill between[of=bandlo and bandhi];

\addplot[black!35, solid, no marks, domain=0.015:4.0, samples=2,
         forget plot] {x};

\addplot[dsurl, mark=*, only marks, mark size=3pt, forget plot]
  table[x=measured_ms, y=predicted_ms, col sep=comma]{fig/data/pvm_url_rows.csv};
\addplot[dsurl, mark=square*, only marks, mark size=2.5pt, forget plot]
  table[x=measured_ms, y=predicted_ms, col sep=comma]{fig/data/pvm_url_nnz.csv};
\addplot[dsurl, mark=triangle*, only marks, mark size=3pt, forget plot]
  table[x=measured_ms, y=predicted_ms, col sep=comma]{fig/data/pvm_url_cyclic.csv};

\addplot[dsnews20, mark=*, only marks, mark size=3pt, forget plot]
  table[x=measured_ms, y=predicted_ms, col sep=comma]{fig/data/pvm_news20_rows.csv};
\addplot[dsnews20, mark=square*, only marks, mark size=2.5pt, forget plot]
  table[x=measured_ms, y=predicted_ms, col sep=comma]{fig/data/pvm_news20_nnz.csv};
\addplot[dsnews20, mark=triangle*, only marks, mark size=3pt, forget plot]
  table[x=measured_ms, y=predicted_ms, col sep=comma]{fig/data/pvm_news20_cyclic.csv};

\addplot[dsrcv1, mark=*, only marks, mark size=3pt, forget plot]
  table[x=measured_ms, y=predicted_ms, col sep=comma]{fig/data/pvm_rcv1_rows.csv};
\addplot[dsrcv1, mark=square*, only marks, mark size=2.5pt, forget plot]
  table[x=measured_ms, y=predicted_ms, col sep=comma]{fig/data/pvm_rcv1_nnz.csv};
\addplot[dsrcv1, mark=triangle*, only marks, mark size=3pt, forget plot]
  table[x=measured_ms, y=predicted_ms, col sep=comma]{fig/data/pvm_rcv1_cyclic.csv};

\addlegendimage{dsurl,   solid, no marks, line width=3pt}
\addlegendentry{url}
\addlegendimage{black, mark=*,        only marks, mark size=2.5pt}
\addlegendentry{rows}
\addlegendimage{dsnews20, solid, no marks, line width=3pt}
\addlegendentry{news20}
\addlegendimage{black, mark=square*,  only marks, mark size=2.5pt}
\addlegendentry{nnz}
\addlegendimage{dsrcv1,  solid, no marks, line width=3pt}
\addlegendentry{rcv1}
\addlegendimage{black, mark=triangle*, only marks, mark size=2.5pt}
\addlegendentry{cyclic}

\end{axis}
\end{tikzpicture}
  \caption{Predicted vs measured per-iteration runtime across the $9$ measured (dataset, partitioner) cells (url $4{\times}64$ at $p{=}256$, news20 $1{\times}64$ at $p{=}64$, rcv1 $1{\times}16$ at $p{=}16$).
  Solid $y{=}x$, shaded band $0.5$ to $2\times$, color = dataset, marker = partitioner.
  Outliers (url rows, nnz, cyclic, news20 rows) reflect a constant dataset-specific \texttt{sparse\_syrkd} overhead.}
  \label{fig:predicted-vs-measured}
\end{figure*}

\subsection{Solver-family transition}
\label{sec:transition}

\Cref{fig:transition} sweeps $p_r$ across all factorizations of $p$ using the cyclic partitioner for each dataset.
We trace the full continuum from 1D $s$-step SGD ($p_r{=}1$) through interior HybridSGD meshes to FedAvg ($p_r{=}p$).
The url panel exhibits a U-shaped per-iteration profile with an empirical minimum at $p_r{=}8$ ($0.557$\,ms/iter).
The topology rule \eqref{eq:meshrule} predicts $p_r{=}4, p_c{=}64$ ($0.606$\,ms/iter), the immediate-neighbor mesh and within $9\%$ of the empirical optimum.
On news20 ($p{=}64$) and rcv1 ($p{=}16$) the rule's $p_c{=}\min(R,p)$ saturates at $p$, so it predicts the 1D $s$-step SGD corner ($p_r{=}1$).
The empirical curves are monotone with the minimum at the same corner on both panels.

\begin{figure*}[t]
  \centering
  \begin{tikzpicture}
\begin{groupplot}[
  group style={
    group size=3 by 1,
    horizontal sep=0.45cm,
  },
  xmode=log, log basis x=2,
  ymode=log,
  xlabel={$p_r$},
  tick label style={font=\footnotesize},
  label style={font=\small},
  title style={font=\small},
  legend style={font=\scriptsize, draw=black!30, fill=white, fill opacity=0.85,
                legend cell align=left},
  every axis plot/.append style={line width=0.8pt, mark size=2.5pt},
  width=0.30\textwidth, height=0.30\textwidth,
  grid=major, grid style={gray!25},
]

\nextgroupplot[
  title={url ($p{=}256$)},
  ylabel={Per-iteration runtime (ms)},
  ymin=0.35, ymax=55,
  xmin=0.7, xmax=350,
  xtick={1,4,16,64,256},
  legend pos=north west,
]
\addplot[lbnnz, solid, mark=square*, mark options={solid,fill=lbnnz!40}]
  table[x=pr, y=per_iter_ms, col sep=comma]{fig/data/trans_url_line.csv};
\addlegendentry{HybridSGD}

\addplot[lbrows, only marks, mark=triangle*, mark options={solid,fill=lbrows!40}]
  table[x=pr, y=per_iter_ms, col sep=comma]{fig/data/trans_url_mbcasgd.csv};
\addlegendentry{1D $s$-step SGD}

\addplot[black, only marks, mark=o, mark options={solid,fill=white}]
  table[x=pr, y=per_iter_ms, col sep=comma]{fig/data/trans_url_dcsgd_s1.csv};
\addlegendentry{FedAvg ($s{=}1$)}

\nextgroupplot[
  title={news20 ($p{=}64$)},
  ymin=0.09, ymax=4.5,
  xmin=0.7, xmax=85,
  xtick={1,4,16,64},
  yticklabels={},
  legend pos=north west,
]
\addplot[lbnnz, solid, mark=square*, mark options={solid,fill=lbnnz!40}]
  table[x=pr, y=per_iter_ms, col sep=comma]{fig/data/trans_news20_line.csv};
\addlegendentry{HybridSGD}

\addplot[lbrows, only marks, mark=triangle*, mark options={solid,fill=lbrows!40}]
  table[x=pr, y=per_iter_ms, col sep=comma]{fig/data/trans_news20_mbcasgd.csv};
\addlegendentry{1D $s$-step SGD}

\addplot[black, only marks, mark=o, mark options={solid,fill=white}]
  table[x=pr, y=per_iter_ms, col sep=comma]{fig/data/trans_news20_dcsgd_s1.csv};
\addlegendentry{FedAvg ($s{=}1$)}

\nextgroupplot[
  title={rcv1 ($p{=}16$)},
  ymin=0.04, ymax=0.25,
  xmin=0.7, xmax=22,
  xtick={1,2,4,8,16},
  yticklabels={},
  legend pos=north west,
]
\addplot[lbnnz, solid, mark=square*, mark options={solid,fill=lbnnz!40}]
  table[x=pr, y=per_iter_ms, col sep=comma]{fig/data/trans_rcv1_line.csv};
\addlegendentry{HybridSGD}

\addplot[lbrows, only marks, mark=triangle*, mark options={solid,fill=lbrows!40}]
  table[x=pr, y=per_iter_ms, col sep=comma]{fig/data/trans_rcv1_mbcasgd.csv};
\addlegendentry{1D $s$-step SGD}

\addplot[black, only marks, mark=o, mark options={solid,fill=white}]
  table[x=pr, y=per_iter_ms, col sep=comma]{fig/data/trans_rcv1_dcsgd_s1.csv};
\addlegendentry{FedAvg ($s{=}1$)}

\end{groupplot}
\end{tikzpicture}
  \caption{Per-iteration runtime vs $p_r$ for three LIBSVM datasets using cyclic partitioner and sweeping all factorizations $p_r p_c = p$.
  Endpoints are 1D $s$-step SGD (left triangle, $p_r{=}1$, $\tau{=}10^4$) and FedAvg (open circle, $s{=}1$).
  url exhibits the U-shape predicted by the topology rule \eqref{eq:meshrule}.
  The empirical minimum lies at $p_r{=}8$ ($0.557$\,ms) and the rule predicts $p_r{=}4, p_c{=}64$ ($0.606$\,ms), the immediate-neighbor mesh and within $9\%$ of the optimum.
  On news20 and rcv1 the rule's $p_c{=}\min(R,p)$ selects $p_c = p$ which correctly predicts 1D $s$-step SGD as the optimum.}
  \label{fig:transition}
\end{figure*}

\subsection{Time-to-target loss}
\label{sec:t-to-target}

We perform offline tuning on the learning rate and set it to $\eta{=}0.01$, which yields convergence on all solvers.
We report single-trial times with inter-trial variance ${\le}\,4\%$ (see \Cref{tab:strong-scaling-headline}) and pick each solver's best $p$.
For HybridSGD, we also pick the best mesh and partitioner.
Per-dataset target losses (\Cref{tab:strong-scaling-headline}) are calibrated to the slower solver's terminal loss within the iteration budget.
This corresponds to FedAvg on url ($0.50$) and news20 ($0.45$).
All solvers reach the same loss on rcv1 ($0.40$). 
Epsilon's loss decays slowly within the iteration budget, so we chose a mid-range value ($0.65$) as the target.
This dataset stresses the convergence rate of the solver.
\Cref{fig:convergence-runtime} shows the convergence behavior of HybridSGD, 1D $s$-step SGD, and FedAvg for each solver's best configuration.
We show the convergence behavior on the url, news20, and rcv1 datasets.
HybridSGD is fastest to the target on url with speedups of $2.8\times$ over 1D $s$-step SGD and $53\times$ over FedAvg.
The speedup over FedAvg arises from two factors.
First, HybridSGD reduces the weights vector Allreduce bandwidth cost to $n/p_c = 101{,}000$ words vs FedAvg's $n = 3{,}231{,}961$).
Second, HybridSGD converges faster since it uses $p_r$ processors instead of $p$ processors in the row dimension (see \Cref{tab:comp-costs} for convergence rate difference).
On rcv1, ties emerge from near-identical trajectories.

\paragraph{Solution quality and the convex objective}
\Cref{eq:logreg} is convex, so every solver minimizes the same objective and, at a comparable step size and a sufficient iteration budget, reaches the same loss (\Cref{fig:convergence-runtime}).
The speedups we report are therefore time-to-equal-loss, not accuracy trade-offs. 
The winning solver reaches a fixed objective value sooner, it does not converge to a different solution.

\begin{figure*}[t]
  \centering
  \begin{tikzpicture}
\begin{groupplot}[
  group style={
    group size=3 by 1,
    horizontal sep=1.1cm,
  },
  ymode=log,
  xlabel={Time (s)},
  tick label style={font=\footnotesize},
  label style={font=\small},
  title style={font=\small},
  legend style={font=\scriptsize, draw=black!30, fill=white, fill opacity=0.85,
                legend cell align=left, legend columns=-1},
  every axis plot/.append style={line width=0.8pt, mark size=2pt},
  width=0.28\textwidth, height=0.30\textwidth,
  grid=major, grid style={gray!25},
]

\nextgroupplot[
  title={url ($p{=}256$)},
  ylabel={Objective (loss)},
  ymin=0.1, ymax=0.75,
  xmin=0, xmax=12.5,
  legend to name=convlegend,
]
\addplot[black, solid, mark=o, mark options={solid,fill=white}]
  table[x=time_s, y=objective, col sep=comma]{fig/data/conv_url_dc.csv};
\addlegendentry{FedAvg}
\addplot[lbnnz, solid, mark=square*, mark options={solid,fill=lbnnz!40}]
  table[x=time_s, y=objective, col sep=comma]{fig/data/conv_url_hyb_cyclic.csv};
\addlegendentry{HybridSGD}
\addplot[lbrows, dashed, mark=triangle*, mark options={solid,fill=lbrows!40}]
  table[x=time_s, y=objective, col sep=comma]{fig/data/conv_url_mbcasgd.csv};
\addlegendentry{1D $s$-step SGD}
\addplot[gray!60, dotted, no marks, line width=0.7pt, domain=0:12.5, samples=2,
         forget plot] {0.50};

\nextgroupplot[
  title={epsilon ($p_{\mathrm{FedAvg}}{=}32,\,p_{\mathrm{Hyb}}{=}512$)},
  ymin=0.62, ymax=0.7,
  xmin=0, xmax=0.7,
  xlabel={},
]
\addplot[black, solid, mark=o, mark options={solid,fill=white}]
  table[x=time_s, y=objective, col sep=comma]{fig/data/conv_epsilon_dc.csv};
\addplot[lbnnz, solid, mark=square*, mark options={solid,fill=lbnnz!40}]
  table[x=time_s, y=objective, col sep=comma]{fig/data/conv_epsilon_hyb_4x128.csv};
\addplot[lbrows, dashed, mark=triangle*, mark options={solid,fill=lbrows!40}]
  table[x=time_s, y=objective, col sep=comma]{fig/data/conv_epsilon_hyb_cyclic.csv};
\addplot[gray!60, dotted, no marks, line width=0.7pt, domain=0:0.7, samples=2,
         forget plot] {0.65};

\nextgroupplot[
  title={rcv1 ($p_{\mathrm{FedAvg}}{=}8,\,p_{\mathrm{Hyb}}{=}16$)},
  ymin=0.35, ymax=0.75,
  xmin=0, xmax=0.17,
]
\addplot[black, solid, mark=o, mark options={solid,fill=white}]
  table[x=time_s, y=objective, col sep=comma]{fig/data/conv_rcv1_dc.csv};
\addplot[lbnnz, solid, mark=square*, mark options={solid,fill=lbnnz!40}]
  table[x=time_s, y=objective, col sep=comma]{fig/data/conv_rcv1_hyb_4x4.csv};
\addplot[lbrows, dashed, mark=triangle*, mark options={solid,fill=lbrows!40}]
  table[x=time_s, y=objective, col sep=comma]{fig/data/conv_rcv1_hyb_cyclic.csv};
\addplot[gray!60, dotted, no marks, line width=0.7pt, domain=0:0.17, samples=2,
         forget plot] {0.40};

\end{groupplot}
\node[anchor=north, yshift=-0.3cm] at (group c2r1.south) {\ref{convlegend}};
\end{tikzpicture}
  \caption{Training loss vs runtime on Perlmutter CPU at each solver's optimal $p$.
  The url dataset uses $p{=}256$ for both solvers, epsilon uses $p_{\mathrm{FedAvg}}{=}32$ and $p_{\mathrm{Hyb}}{=}512$, and rcv1 uses $p_{\mathrm{FedAvg}}{=}8$ and $p_{\mathrm{Hyb}}{=}16$.
  HybridSGD uses cyclic partitioner.
  Dotted horizontal line on each panel marks the time-to-target loss threshold used in \Cref{tab:strong-scaling-headline} ($0.50$ for url, $0.65$ for epsilon, $0.40$ for rcv1).
  Trajectories continue past the threshold because each solver runs for a fixed iteration budget rather than stopping at target.
  On url, FedAvg takes $\sim$10\,s to reach loss $0.45$, while HybridSGD reaches $0.25$ in $\sim$1\,s.
  On epsilon, FedAvg descends faster initially.
  On rcv1, all solvers converge in comparable runtime.}
  \label{fig:convergence-runtime}
\end{figure*}

\begin{table}[t]
  \centering
  \caption{Time-to-target loss on Perlmutter CPU ($64$ ranks/node $\times$ $2$ cores/task, FP64, $\eta{=}0.01$, $b{=}32$, $s{=}4$, $\tau{=}10$).
  ``Best FedAvg'' picks FedAvg's fastest configuration over $p$, while ``Best HybridSGD'' picks HybridSGD's fastest over $p$, mesh, and partitioner.
  We performed five trials and report the standard deviation as percentage of the mean for HybridSGD.
  Speedups above $1$ indicate a HybridSGD advantage.
  The dense epsilon dataset falls in the compute-dominated regime, where FedAvg's cheaper per-iteration computation outweighs HybridSGD's communication savings; this regime boundary is predicted by the cost model.
}
  \label{tab:strong-scaling-headline}
  \footnotesize
  \resizebox{\textwidth}{!}{%
  \begin{tabular}{lcrlrr}
    \hline
    Dataset & Target & Best FedAvg & Best HybridSGD & Variance (HybridSGD) & Speedup \\
            & (loss)  & ($p$, time) & (mesh, $p$, partitioner, time) & & (Hyb/FedAvg) \\
    \hline
    url     & $0.50$ & $256$,~$9.48$\,s & $8{\times}32$,~$256$, cyclic, $\mathbf{0.179}$\,s & $\pm 2.6\%$ & $\mathbf{53.0\times}$ \\
    news20  & $0.45$ & $8$,~$6.80$\,s   & $1{\times}64$,~$64$,  cyclic, $\mathbf{0.465}$\,s & $\pm 1.7\%$ & $\mathbf{14.6\times}$ \\
    rcv1    & $0.40$ & $8$,~$0.156$\,s  & $1{\times}16$,~$16$,  cyclic, $0.141$\,s          & $\pm 1.4\%$ & $1.11\times$ \\
    epsilon & $0.65$ & $32$,~$0.20$\,s  & $1{\times}512$,~$512$ (dense, partitioner irrelevant), $0.449$\,s & $\pm 0.8\%$ & $0.44\times$ (FedAvg $2.25\times$) \\
    \hline
  \end{tabular}%
  }
\end{table}

\Cref{fig:strong-scaling-url} (left) shows per-iteration speedup on url as a function of $p$ for FedAvg, HybridSGD $1{\times}p$, and HybridSGD $8{\times}(p/8)$ using the cyclic partitioner.
FedAvg's per-iteration time is flat at ${\approx}1\times$ since increasing $p$ also increases work per rank.
Setting $\tau{=}10$ keeps the AllReduce overhead below $1\%$ of per-iteration cost, so the communication bottleneck does not appear in the FedAvg experiment.
HybridSGD $1{\times}p$ is similarly flat near $1\times$.
This bottleneck is due to url's extreme column skew and is not a communication bottleneck from Gram Allreduce size.
HybridSGD $8{\times}(p/8)$ reaches $5.7\times$ at $p{=}1024$.
This is achieved through the use of a smaller row team ($p_c{=}p/8$) which shrinks the weights vector and the Gram Allreduce size while keeping the row team intra-node through $p{=}512$.

\Cref{fig:strong-scaling-url} (right) shows per-iteration speedup on a synthetic uniform nnz distributed sparse matrix.
We observe speedups from 1D $s$-step SGD when column skew is removed.
HybridSGD continues to show greater per-iteration speedup than the other solvers due to the additional mesh dimension, which allows for finding a better Allreduce bandwidth balance point than the extremes of FedAvg and 1D $s$-step SGD.

In both experiments, we observe superlinear speedups at $p = \{128, 256\}$ due to caching effects, where cache size grows while local problem size decreases.
This superlinear speedup disappears once communication becomes a larger fraction of runtime.
\begin{figure*}[t]
  \centering
  \begin{tikzpicture}
\begin{axis}[
  plotstyle,
  width=0.48\textwidth, height=0.40\textwidth,
  xmode=log, ymode=log,
  log basis x=2,
  log basis y=2,
  xlabel={$p$ (ranks)},
  ylabel={Speedup over $p{=}64$},
  xmin=48, xmax=5500,
  ymin=0.4, ymax=72,
  xtick={64,128,256,512,1024,2048,4096},
  xticklabels={64,128,256,512,1024,2048,4096},
  ytick={0.5,1,2,4,8,16,32,64},
  yticklabels={$\frac{1}{2}$,1,2,4,8,16,32,64},
  tick label style={font=\footnotesize},
  legend pos=north west,
  legend cell align=left,
]

\addplot[black!30, dashed, no marks, domain=64:4096, samples=60] {x/64};
\addlegendentry{ideal}

\addplot[black, solid, mark=o, mark options={solid,fill=white}]
  coordinates {
    (64,   1.0000)
    (128,  1.0165)
    (256,  1.0031)
    (512,  1.0175)
    (1024, 1.0137)
    (2048, 1.0109)
    (4096, 1.0032)
  };
\addlegendentry{FedAvg}

\addplot[lbrows, solid, mark=triangle*, mark options={solid,fill=lbrows!40}]
  coordinates {
    (64,   1.0000)
    (128,  1.0534)
    (256,  1.0543)
    (512,  0.9236)
    (1024, 0.9278)
    (2048, 0.8355)
    (4096, 0.8769)
  };
\addlegendentry{1D $s$-step SGD}

\addplot[lbnnz, solid, mark=square*, mark options={solid,fill=lbnnz!40}]
  coordinates {
    (64,   1.0000)
    (128,  2.4061)
    (256,  4.1817)
    (512,  5.2609)
    (1024, 5.6826)
    (2048, 4.9630)
    (4096, 4.8366)
  };
\addlegendentry{Hyb $8{\times}(p/8)$}

\end{axis}
\end{tikzpicture}%
  \hfill
  \begin{tikzpicture}
\begin{axis}[
  plotstyle,
  width=0.48\textwidth, height=0.40\textwidth,
  xmode=log, ymode=log,
  log basis x=2,
  log basis y=2,
  xlabel={$p$ (ranks)},
  ylabel={Speedup over $p{=}64$},
  xmin=48, xmax=2800,
  ymin=0.7, ymax=48,
  xtick={64,128,256,512,1024,2048},
  xticklabels={64,128,256,512,1024,2048},
  ytick={1,2,4,8,16,32},
  yticklabels={1,2,4,8,16,32},
  tick label style={font=\footnotesize},
  legend pos=north west,
  legend cell align=left,
]

\addplot[black!30, dashed, no marks, domain=64:2048, samples=60] {x/64};
\addlegendentry{ideal}

\addplot[black, solid, mark=o, mark options={solid,fill=white}]
  coordinates {
    (64,  1.000)
    (128, 1.003)
    (256, 1.003)
    (512, 1.003)
    (1024,0.981)
    (2048,0.976)
  };
\addlegendentry{FedAvg}

\addplot[lbrows, solid, mark=triangle*, mark options={solid,fill=lbrows!40}]
  coordinates {
    (64,  1.000)
    (128, 1.386)
    (256, 1.651)
    (512, 1.921)
    (1024,3.413)
    (2048,3.411)
  };
\addlegendentry{1D $s$-step SGD}

\addplot[lbnnz, solid, mark=square*, mark options={solid,fill=lbnnz!40}]
  coordinates {
    (64,  1.000)
    (128, 2.294)
    (256, 4.312)
    (512, 7.353)
    (1024,11.053)
    (2048,10.673)
  };
\addlegendentry{Hyb $4{\times}(p/4)$}

\end{axis}
\end{tikzpicture}
  \caption{\emph{Left}: Per-iteration speedup on url ($b{=}32$, $s{=}4$, $\tau{=}10$, $\eta{=}0.01$, cyclic partitioner).
  FedAvg and 1D $s$-step SGD are flat at ${\approx}1\times$.
  Hyb $8{\times}(p/8)$ reaches $5.7\times$ at $p{=}1024$ by shrinking the row-weight and Gram Allreduce message sizes.
  \emph{Right}: Per-iteration speedup on a synthetic uniform sparse matrix ($m{=}2^{21}$, $n{=}3{,}145{,}728$, density $0.4\%$).
  This isolates solver performance from column skew.
  FedAvg is flat at ${\approx}1\times$ because the full-$n$ Allreduce ($25$\,MB) dominates regardless of $p$.
  HybridSGD $4{\times}(p/4)$ reaches $11.1\times$ at $p{=}1024$ ($69\%$ efficiency) and $10.7\times$ at $p{=}2048$ ($33\%$ efficiency).
  Superlinear speedup reflects cache effects as $p$ increases and as the local subproblem shrinks.}
  \label{fig:strong-scaling-url}
\end{figure*}

\section{Conclusion}
HybridSGD is a 2D-parallel SGD method that generalizes 1D $s$-step SGD and 1D FedAvg into a continuous family indexed by the mesh split $p = p_r \times p_c$, recovering each 1D baseline at an extreme grid dimension and exposing the interior meshes as a productive design space.
The two 1D approaches were previously seen as incompatible because each fixes a single data partitioning axis.
We show that the partitionings are in fact compatible because they live on orthogonal axes of the 2D mesh.
This compatibility lets a single solver navigate the trade-off between the recurrence-unrolling regime of $s$-step SGD and the deferred-averaging regime of FedAvg by varying $(p_r, p_c)$.
The framework is implemented in C++/MPI on Intel MKL sparse BLAS with three selectable column partitioners (rows, nonzero-greedy, cyclic) and a configurable mesh, exposing the full $(p_r, p_c, s, b, \tau, \text{partitioner})$ design space for evaluation.

We present a closed-form $\alpha$-$\beta$-$\gamma$ cost model that decomposes the per-epoch wall into compute, latency, Gram-bandwidth, and sync-bandwidth contributions, with optima for the recurrence length $s$ and mini-batch size $b$ and a bandwidth balance condition $(s{-}1)sb^2\tau p_c \approx 2n$ that distinguishes the Gram-bandwidth and sync-bandwidth regimes.
We also present a cache-aware compute parameter $\gamma(W)$, a rank-aware bandwidth parameter $\beta(p)$ that captures the intra-to-inter-node transition, and a nonzero-imbalance multiplier $\kappa$ to capture load-imbalance and characterize slow ranks.
We present a parameter-free topology rule, $p_c^{*} = \max(\lceil nw/L_{\text{cap}}\rceil,\,\min(R, p))$, which sets the mesh split from the per-node rank count $R$ and the per-core cache capacity $L_{\text{cap}}$ alone.
The rule reduces mesh selection to two machine constants and removes the transcendental fixed-point that the analytic optimum would otherwise require.

Cache-aware partitioning is a two-objective constrained problem.
We minimize the nonzero imbalance $\kappa$ subject to a per-rank cache-footprint capacity.
A nonzero-balanced contiguous partitioner achieves $\kappa \approx 1$ at the cost of overloading the rank that absorbs the heavy columns, which spills out of L2 into DRAM and degrades per-iteration runtime by up to $2.4\times$ on column-skewed data.
A cyclic partitioner satisfies both objectives in expectation, bounding $n_{\text{local}}$ exactly to $n/p_c$ while keeping $\kappa$ near $1$, at the cost of a column permutation in the reader.
The cost model predicts and measurement confirms that no single partitioner dominates across the LIBSVM suite.
The appropriate choice is dataset-conditional and is a first-class HybridSGD parameter.

At the strong-scaling limit on LIBSVM, HybridSGD attains $\mathbf{53\times}$ time-to-target-loss speedup on url and $\mathbf{14.6\times}$ on news20 over FedAvg, matches FedAvg on rcv1, and is outperformed by FedAvg on dense epsilon where the cheaper per-iteration compute dominates.
The qualitative regime crossover at each dataset is predicted by the cost model and confirmed empirically.
These results position HybridSGD as a regime-aware solver, most beneficial in sparse, high-dimensional, communication- and skew-limited settings, rather than a universal replacement for FedAvg.
The speedups decompose into a per-iteration throughput component (driven by the smaller column-Allreduce payload of $n/p_c$ versus the full-$n$ payload of FedAvg) and a sample-efficiency component (driven by the smaller effective per-sync batch satisfying Stich's bound at large $p$, where FedAvg's local-update drift inflates iterations-to-loss).
The framework extends directly to other convex losses for which $s$-step methods apply, with the cost model and topology rule transferring unchanged.

\section*{Acknowledgements}
Thanks to Grey Ballard for numerous, helpful discussions on this work.
This work was supported by the U.S. Department of Energy, Office of Science, Advanced Scientific Computing Research (ASCR) program under Award Number DE-SC-0023296.
This research used resources of the National Energy Research Scientific Computing Center (NERSC), a Department of Energy Office of Science User Facility using NERSC award ASCR-ERCAP0026261 and ASCR-ERCAP0028617.

\bibliographystyle{abbrv}
\bibliography{refs}
\end{document}